\documentclass[a4paper,UKenglish,cleveref, autoref, thm-restate]{lipics-v2021}

\hideLIPIcs  

\title{Completing the Complexity Classification of 2-Solo Chess: Knights and Kings are Hard} 

\titlerunning{2-Solo Chess: Knights and Kings Are Hard} 

\author{Kolja Kühn}{Institute of Theoretical Informatics, Karlsruhe Institute of Technology, Germany}{kolja.kuehn@kit.edu}{https://orcid.org/0009-0006-0746-6876}{}

\author{Wendy Yi}{Institute of Theoretical Informatics, Karlsruhe Institute of Technology, Germany}{wendy.yi@kit.edu}{https://orcid.org/0009-0006-3025-4505}{}

\authorrunning{K. Kühn and W. Yi} 

\Copyright{Kolja Kühn and Wendy Yi} 



\category{} 

\relatedversiondetails{Conference Version}{https://doi.org/10.4230/LIPIcs.FUN.2026.27} 

\nolinenumbers 

\EventEditors{John Iacono}
\EventNoEds{1}
\EventLongTitle{13th International Conference on Fun with Algorithms (FUN 2026)}
\EventShortTitle{FUN 2026}
\EventAcronym{FUN}
\EventYear{2026}
\EventDate{May 18--22, 2026}
\EventLocation{Porquerolles, France}
\EventLogo{}
\SeriesVolume{366}
\ArticleNo{6}

\usepackage{tikz} 
\usetikzlibrary{shapes.geometric}
\usepackage{xspace}
\usepackage{todonotes}
\usepackage{mathtools}
\usepackage{booktabs}
\usepackage{complexity}

\begin{document}

\newcommand{\twoSquareMove}[2]{(#1 \rightarrow #2)}
\newcommand{\threeSquareMove}[3]{(#1 \rightarrow #2 \rightarrow #3)}

\newcommand{\solo}{\textsc{Solo Chess}\xspace}
\newcommand{\sat}{\textsc{3,3-SAT}\xspace}
\newcommand{\checking}{1-TEST\xspace}

\newcommand{\tileSize}{100\xspace}
\newcommand{\red}[1]{\textcolor{red}{#1}}

\DeclarePairedDelimiterX\set[1]{\lbrace}{\rbrace}{\def\given{\;\delimsize\vert\;}#1}

\newcounter{numalphcounter}
\newcommand{\numalph}[1]{\setcounter{numalphcounter}{#1}\Alph{numalphcounter}}
\usetikzlibrary{matrix, positioning}

\tikzset{move/.style={
line width=0.35mm, 
-latex,
shorten <= 0.3cm,
 shorten >= 0.3cm,
 rounded corners=0.5cm,
 black,
 xshift=0.5cm,
 yshift=-0.5cm,
}}

\tikzset{edge/.style={
line width=0.35mm,
shorten <= 0.3cm,
 shorten >= 0.3cm,
 rounded corners=0.5cm,
 black,
 xshift=0.5cm,
 yshift=-0.5cm,
}}

\tikzset{halfEdge/.style={
line width=0.35mm,
dashed,
shorten <= 0.3cm,
 shorten >= 1.1cm,
 rounded corners=0.5cm,
 black,
 xshift=0.5cm,
 yshift=-0.5cm,
}}

\definecolor{dark}{HTML}{b58863}
\definecolor{light}{HTML}{f0d9b5}

\newcommand{\board}[2]{
    \foreach \y in {1,...,#1}{
        \foreach \x in {1,...,#2}{
            \ifodd\numexpr\x+\y\relax
                \fill[light] (\x,\y) rectangle ($(\x+1,\y+1)$);
            \else
                \fill[dark] (\x,\y) rectangle ($(\x+1,\y+1)$);
            \fi
        }
    }
}

\newcommand{\squareColor}[3]{
    \fill[#3] (#1,#2) rectangle ($(#1+1,#2+1)$);
}

\newcommand{\splitSquareColor}[4]{
    \fill[#3] (#1,#2) rectangle ($(#1+0.5,#2+1)$);
    \fill[#4] ($(#1+0.5,#2)$) rectangle ($(#1+1,#2+1)$);
}

\newcommand{\knight}[4][]{
    \node (tmp) at ($(#2+0.5,#3+0.5)$) {\includegraphics[width=0.85cm]{images/knight/knight_new_#4.pdf}};
    \ifthenelse{\equal{#1}{}}{}{
        \node[inner sep=0cm, outer sep=0cm, fill=white, circle, minimum width=0.3cm, fill opacity=0.7, text opacity=1] at (tmp.center) {\numalph{#1}};
    }
}

\newcommand{\knightOrange}[4][]{
    \squareColor{#2}{#3}{kit-yellow}
    \knight[#1]{#2}{#3}{#4}
}

\newcommand{\knightBlue}[4][]{
    \squareColor{#2}{#3}{kit-blue}
    \knight[#1]{#2}{#3}{#4}
}

\newcommand{\knightSplit}[4][]{
    \splitSquareColor{#2}{#3}{kit-yellow}{kit-blue}
    \knight[#1]{#2}{#3}{#4}
}

\newcommand{\knightSignals}{
    \begin{tabular}{@{}c@{}}
        \begin{tikzpicture}[x=1cm,y=-1cm]
            \board{4}{7}
            \path[use as bounding box] (1,1) rectangle (8,5);
            \foreach \x/\y/\c [count=\i] in {
                1/1/2, 1/3/0, 3/2/2, 3/4/2, 5/1/2, 5/3/2, 7/2/2, 7/4/2%
            }{
                \knight[\i]{\x}{\y}{\c}
            }
    
            \draw[line width=0.35mm] (1.83,1.7) -- (3.17,2.3);
            \draw[line width=0.35mm] (1.83,3.3) -- (3.17,2.7);
            \draw[line width=0.35mm] (1.83,3.7) -- (3.17,4.3);
            \draw[line width=0.35mm] (3.83,2.3) -- (5.17,1.7);
            \draw[line width=0.35mm] (3.83,2.7) -- (5.17,3.3);
            \draw[line width=0.35mm] (3.83,4.3) -- (5.17,3.7);
            
            \draw[line width=0.35mm] (5.83,1.7) -- (7.17,2.3);
            \draw[line width=0.35mm] (5.83,3.3) -- (7.17,2.7);
            \draw[line width=0.35mm] (5.83,3.7) -- (7.17,4.3);
            
            \draw[line width=0.35mm, dashed, shorten >= 0.7cm] (7.83,2.3) -- (9.17,1.7);
            \draw[line width=0.35mm, dashed, shorten >= 0.7cm] (7.83,2.7) -- (9.17,3.3);
            \draw[line width=0.35mm, dashed, shorten >= 0.7cm] (7.83,4.3) -- (9.17,3.7);
            
            \draw [rounded corners = 0.5cm, line width=0.5mm, -latex, blue] (3.25,4.2) -- (1.65,3.5) -- (3.25,2.8);
            \draw [rounded corners = 0.5cm, line width=0.5mm, -latex, blue] (1.65,1.88) -- (5.11,3.55);
        \end{tikzpicture}
        
        \\
        
        \begin{tikzpicture}[x=1cm,y=-1cm]
            \board{4}{7}
            \path[use as bounding box] (1,1) rectangle (8,5);
            \foreach \x/\y/\c [count=\i from 2] in {
                1/3/1, 3/2/2, 3/4/2, 5/1/2, 5/3/2, 7/2/2, 7/4/2%
            }{
                \knight[\i]{\x}{\y}{\c}
            }

            \draw[line width=0.35mm] (1.83,3.3) -- (3.17,2.7);
            \draw[line width=0.35mm] (1.83,3.7) -- (3.17,4.3);
            \draw[line width=0.35mm] (3.83,2.3) -- (5.17,1.7);
            \draw[line width=0.35mm] (3.83,2.7) -- (5.17,3.3);
            \draw[line width=0.35mm] (3.83,4.3) -- (5.17,3.7);
            
            \draw[line width=0.35mm] (5.83,1.7) -- (7.17,2.3);
            \draw[line width=0.35mm] (5.83,3.3) -- (7.17,2.7);
            \draw[line width=0.35mm] (5.83,3.7) -- (7.17,4.3);
            
            \draw[line width=0.35mm, dashed, shorten >= 0.7cm] (7.83,2.3) -- (9.17,1.7);
            \draw[line width=0.35mm, dashed, shorten >= 0.7cm] (7.83,2.7) -- (9.17,3.3);
            \draw[line width=0.35mm, dashed, shorten >= 0.7cm] (7.83,4.3) -- (9.17,3.7);
            
            \draw [rounded corners = 0.5cm, line width=0.5mm, -latex, blue] (1.88,3.4) -- (3.25,2.8);
            \draw [rounded corners = 0.5cm, line width=0.5mm, -latex, blue] (5.25,1.8) -- (3.65,2.5) -- (5.25,3.2);
            \draw [rounded corners = 0.5cm, line width=0.5mm, -latex, blue] (3.88,4.4) -- (5.25,3.8);
        \end{tikzpicture}
    \end{tabular}
}

\newcommand{\knightWireZeroed}{
    \begin{tikzpicture}[x=1cm,y=-1cm]
        \board{4}{11}
        \path[use as bounding box] (1,1) rectangle (12,5);
        \foreach \b/\x/\y [count=\i] in {
            2/1/1, 0/1/3, 0/3/2, 2/3/4, 2/5/1, 0/5/3, 0/7/2, 2/7/4, 2/9/1, 0/9/3, 0/11/2, 2/11/4
        }{
            \knight[\i]{\x}{\y}{\b}
        }

        \foreach \u/\v/\x/\y in {
            1/1/3/2, 1/3/3/2, 1/3/3/4, 3/2/5/1, 3/2/5/3, 3/4/5/3, 5/1/7/2, 5/3/7/2, 5/3/7/4, 7/2/9/1, 7/2/9/3, 7/4/9/3, 9/1/11/2, 9/3/11/2, 9/3/11/4
        }{
            \draw[edge] (\u,\v) -- (\x,\y);
        }
        \foreach \u/\v/\x/\y in {
            1/1/-1/2, 1/3/-1/2, 1/3/-1/4, 11/2/13/1, 11/2/13/3, 11/4/13/3
        }{
            \draw[halfEdge] (\u,\v) -- (\x,\y);
        }
    \end{tikzpicture}
}

\newcommand{\knightAfterSignals}{
    \begin{tabular}{@{}c@{}}
        \begin{tikzpicture}[x=1cm,y=-1cm]
            \board{4}{7}
            \path[use as bounding box] (1,1) rectangle (8,5);
            \foreach \x/\y/\c [count=\i from 5] in {
                5/1/2, 5/3/0, 7/2/2, 7/4/2%
            }{
                \knight[\i]{\x}{\y}{\c}
            }
            
            \draw[line width=0.35mm] (5.83,1.7) -- (7.17,2.3);
            \draw[line width=0.35mm] (5.83,3.3) -- (7.17,2.7);
            \draw[line width=0.35mm] (5.83,3.7) -- (7.17,4.3);
            
            \draw[line width=0.35mm, dashed, shorten >= 0.7cm] (7.83,2.3) -- (9.17,1.7);
            \draw[line width=0.35mm, dashed, shorten >= 0.7cm] (7.83,2.7) -- (9.17,3.3);
            \draw[line width=0.35mm, dashed, shorten >= 0.7cm] (7.83,4.3) -- (9.17,3.7);
            
            \draw [rounded corners = 0.5cm, line width=0.35mm, -latex, blue] (3.25,4.2) -- (1.65,3.5) -- (3.25,2.8);
            \draw [rounded corners = 0.5cm, line width=0.35mm, -latex, blue] (1.65,1.88) -- (5.11,3.55);
        \end{tikzpicture}

        \\

        \begin{tikzpicture}[x=1cm,y=-1cm]
            \board{4}{7}
            \path[use as bounding box] (1,1) rectangle (8,5);
            \foreach \x/\y/\c [count=\i from 6] in {
                5/3/1, 7/2/2, 7/4/2%
            }{
                \knight[\i]{\x}{\y}{\c}
            }
            
            \draw[line width=0.35mm] (5.83,3.3) -- (7.17,2.7);
            \draw[line width=0.35mm] (5.83,3.7) -- (7.17,4.3);
            
            \draw[line width=0.35mm, dashed, shorten >= 0.7cm] (7.83,2.3) -- (9.17,1.7);
            \draw[line width=0.35mm, dashed, shorten >= 0.7cm] (7.83,2.7) -- (9.17,3.3);
            \draw[line width=0.35mm, dashed, shorten >= 0.7cm] (7.83,4.3) -- (9.17,3.7);
            
            \draw[line width=0.35mm, dashed, shorten >= 0.7cm] (7.83,2.3) -- (9.17,1.7);
            \draw[line width=0.35mm, dashed, shorten >= 0.7cm] (7.83,2.7) -- (9.17,3.3);
            \draw[line width=0.35mm, dashed, shorten >= 0.7cm] (7.83,4.3) -- (9.17,3.7);
            
            \draw [rounded corners = 0.5cm, line width=0.35mm, -latex, blue] (1.88,3.4) -- (3.25,2.8);
            \draw [rounded corners = 0.5cm, line width=0.35mm, -latex, blue] (5.25,1.8) -- (3.65,2.5) -- (5.25,3.2);
            \draw [rounded corners = 0.5cm, line width=0.35mm, -latex, blue] (3.88,4.4) -- (5.25,3.8);
        \end{tikzpicture}
    \end{tabular}
}

\newcommand{\knightWireCorner}{
    \begin{tikzpicture}[x=1cm,y=-1cm]
        \board{8}{9}
        \path[use as bounding box] (1,1) rectangle (8.5, 9.5);
        \foreach \x/\y [count=\i] in {
            1/8, 3/8, 2/6, 4/6, 4/5, 3/4, 5/3, 5/5, 7/2, 7/4, 9/3, 9/5
        }{
            \knight{\x}{\y}{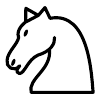}
        }
        \foreach \u/\v/\x/\y in {
            1/8/2/6, 3/8/2/6, 3/8/4/6, 2/6/3/4, 2/6/4/5, 4/6/3/4, 4/5/5/3, 3/4/5/3, 3/4/5/5, 5/3/7/2, 5/3/7/4, 5/5/7/4, 7/2/9/3, 7/4/9/3, 7/4/9/5
        }{
            \draw[edge] (\u,\v) -- (\x,\y);
        }
        \foreach \u/\v/\x/\y in {
            9/3/11/2, 9/3/11/4, 9/5/11/4, 1/8/2/10, 3/8/2/10, 3/8/4/10
        }{
            \draw[halfEdge] (\u,\v) -- (\x,\y);
        }
    \end{tikzpicture}
}

\newcommand{\knightWireDiagonalForward}{
    \begin{tikzpicture}[x=1cm,y=-1cm]
        \board{5}{10}
        \path[use as bounding box] (0.5,1) rectangle (11.5,6);
        \fill[gray] (1,1) rectangle (2,2);
        \fill[gray] (1,3) rectangle (2,4);
        \fill[gray] (5,1) rectangle (6,2);
        \fill[gray] (5,3) rectangle (6,4);
        \fill[gray] (9,1) rectangle (10,2);
        \fill[gray] (9,3) rectangle (10,4);
        \fill[gray!50] (10,2) rectangle (11,3);
        \fill[gray!50] (10,4) rectangle (11,5);
        
        \foreach \x/\y [count=\i] in {
            1/1, 1/3, 3/2, 3/4, 5/1, 5/3, 7/2, 4/5, 9/1, 6/4, 8/3, 8/5, 10/2, 10/4%
        }{
            \knight{\x}{\y}{2}
        }

        \draw[line width=0.35mm] (1.83,1.7) -- (3.17,2.3);
        \draw[line width=0.35mm] (1.83,3.3) -- (3.17,2.7);
        \draw[line width=0.35mm] (1.83,3.7) -- (3.17,4.3);
        \draw[line width=0.35mm] (3.83,2.3) -- (5.17,1.7);
        \draw[line width=0.35mm] (3.83,2.7) -- (5.17,3.3);
        \draw[line width=0.35mm] (3.83,4.3) -- (5.17,3.7);
        
        \draw[line width=0.35mm] (5.83,1.7) -- (7.17,2.3);
        \draw[line width=0.35mm] (5.83,3.3) -- (7.17,2.7);
        \draw[line width=0.35mm] (7.83,2.3) -- (9.17,1.7);
        
        \draw[line width=0.35mm] (5.3,3.83) -- (4.7,5.17);
        \draw[line width=0.35mm] (4.83,5.3) -- (6.17,4.7);
        \draw[line width=0.35mm] (7.3,2.83) -- (6.7,4.17);
        \draw[line width=0.35mm] (9.3,1.83) -- (8.7,3.17);
        
        \draw[line width=0.35mm] (6.83,4.3) -- (8.17,3.7);
        \draw[line width=0.35mm] (6.83,4.7) -- (8.17,5.3);
        \draw[line width=0.35mm] (8.83,3.3) -- (10.17,2.7);
        \draw[line width=0.35mm] (8.83,3.7) -- (10.17,4.3);
        \draw[line width=0.35mm] (8.83,5.3) -- (10.17,4.7);

        \draw[line width=0.35mm, dashed, shorten >= 0.7cm] (1.17,1.7) -- (-0.17,2.3);
        \draw[line width=0.35mm, dashed, shorten >= 0.7cm] (1.17,3.3) -- (-0.17,2.7);
        \draw[line width=0.35mm, dashed, shorten >= 0.7cm] (1.17,3.7) -- (-0.17,4.3);
        
        \draw[line width=0.35mm, dashed, shorten >= 0.7cm] (10.83,2.7) -- (12.17,3.3);
        \draw[line width=0.35mm, dashed, shorten >= 0.7cm] (10.83,4.3) -- (12.17,3.7);
        \draw[line width=0.35mm, dashed, shorten >= 0.7cm] (10.83,4.7) -- (12.17,5.3);

        \draw [line width=0.35mm, -latex] (9.8, 1.8) -- (10.2,2.2);
        \draw [line width=0.35mm, -latex] (9.8, 3.8) -- (10.2,4.2);
    \end{tikzpicture}
}

\newcommand{\knightWireDiagonalBackward}{
    \begin{tikzpicture}[x=1cm,y=-1cm]
        \board{5}{9}
        \path[use as bounding box] (0.5,1) rectangle (10.2,6);
        \fill[gray] (1,1) rectangle (2,2);
        \fill[gray] (1,3) rectangle (2,4);
        \fill[gray] (5,1) rectangle (6,2);
        \fill[gray] (5,3) rectangle (6,4);
        \fill[gray] (9,1) rectangle (10,2);
        \fill[gray] (9,3) rectangle (10,4);
        \fill[gray!50] (8,2) rectangle (9,3);
        \fill[gray!50] (8,4) rectangle (9,5);
        
        \foreach \x/\y [count=\i] in {
            1/1, 1/3, 3/2, 2/5, 5/1, 4/4, 6/3, 6/5, 8/2, 8/4%
        }{
            \knight{\x}{\y}{2}
        }

        \draw[line width=0.35mm] (1.83,1.7) -- (3.17,2.3);
        \draw[line width=0.35mm] (1.83,3.3) -- (3.17,2.7);
        \draw[line width=0.35mm] (1.7,3.83) -- (2.3,5.17);
        \draw[line width=0.35mm] (3.83,2.3) -- (5.17,1.7);
        \draw[line width=0.35mm] (3.7,2.83) -- (4.3,4.17);
        \draw[line width=0.35mm] (2.83,5.3) -- (4.17,4.7);

        \draw[line width=0.35mm] (5.7,1.83) -- (6.3,3.17);
        \draw[line width=0.35mm] (4.83,4.3) -- (6.17,3.7);
        \draw[line width=0.35mm] (4.83,4.7) -- (6.17,5.3);
        
        \draw[line width=0.35mm] (6.83,3.3) -- (8.17,2.7);
        \draw[line width=0.35mm] (6.83,3.7) -- (8.17,4.3);
        \draw[line width=0.35mm] (6.83,5.3) -- (8.17,4.7);

        \draw[line width=0.35mm, dashed, shorten >= 0.7cm] (1.17,1.7) -- (-0.17,2.3);
        \draw[line width=0.35mm, dashed, shorten >= 0.7cm] (1.17,3.3) -- (-0.17,2.7);
        \draw[line width=0.35mm, dashed, shorten >= 0.7cm] (1.17,3.7) -- (-0.17,4.3);
        
        \draw[line width=0.35mm, dashed] (8.83,2.7) -- (10.17,3.3);
        \draw[line width=0.35mm, dashed] (8.83,4.3) -- (10.17,3.7);
        \draw[line width=0.35mm, dashed] (8.83,4.7) -- (10.17,5.3);
        
        \draw [line width=0.35mm, -latex] (9.2, 1.8) -- (8.8,2.2);
        \draw [line width=0.35mm, -latex] (9.2, 3.8) -- (8.8,4.2);
    \end{tikzpicture}
}

\newcommand{\knightRotatedWire}{
    \begin{tikzpicture}[x=1cm,y=-1cm]
        \board{6}{12}
        
        \foreach \x/\y in {
            1/1, 1/3, 3/2, 2/5, 5/1, 4/4, 6/3, 5/6, 7/1, 7/5, 8/3, 9/4, 10/4, 10/2, 12/3, 12/1
        }{
            \knight{\x}{\y}{2}
        }

        \draw[line width=0.35mm] (1.83,1.7) -- (3.17,2.3);
        \draw[line width=0.35mm] (1.83,3.3) -- (3.17,2.7);
        \draw[line width=0.35mm] (1.7,3.83) -- (2.3,5.17);
        \draw[line width=0.35mm] (3.83,2.3) -- (5.17,1.7);
        \draw[line width=0.35mm] (3.7,2.83) -- (4.3,4.17);
        \draw[line width=0.35mm] (2.83,5.3) -- (4.17,4.7);

        \draw[line width=0.35mm] (5.7,1.83) -- (6.3,3.17);
        \draw[line width=0.35mm] (4.83,4.3) -- (6.17,3.7);
        \draw[line width=0.35mm] (4.7,4.83) -- (5.3,6.17);
        \draw[line width=0.35mm] (6.7,3.83) -- (7.3,5.17);
        \draw[line width=0.35mm] (5.83,6.3) -- (7.17,5.7);
        
        \draw[line width=0.35mm] (6.7,3.17) -- (7.3,1.83);
        \draw[line width=0.35mm] (7.7,1.83) -- (8.3,3.17);
        \draw[line width=0.35mm] (7.7,5.17) -- (8.3,3.83);
        
        \draw[line width=0.35mm] (8.83,3.3) -- (10.17,2.7);
        \draw[line width=0.35mm] (8.83,3.7) -- (10.17,4.3);
        \draw[line width=0.35mm] (7.83,5.3) -- (9.17,4.7);
        \draw[line width=0.35mm] (9.7,4.17) -- (10.3,2.83);
        
        \draw[line width=0.35mm] (10.83,2.3) -- (12.17,1.7);
        \draw[line width=0.35mm] (10.83,2.7) -- (12.17,3.3);
        \draw[line width=0.35mm] (10.83,4.3) -- (12.17,3.7);

        \draw[line width=0.35mm, dashed, shorten >= 0.7cm] (1.17,1.7) -- (-0.17,2.3);
        \draw[line width=0.35mm, dashed, shorten >= 0.7cm] (1.17,3.3) -- (-0.17,2.7);
        \draw[line width=0.35mm, dashed, shorten >= 0.7cm] (1.17,3.7) -- (-0.17,4.3);
        
        \draw[line width=0.35mm, dashed] (12.83,1.7) -- (14.17,2.3);
        \draw[line width=0.35mm, dashed] (12.83,3.3) -- (14.17,2.7);
        \draw[line width=0.35mm, dashed] (12.83,3.7) -- (14.17,4.3);
    \end{tikzpicture}
}

\newcommand{\knightHighWire}{
    \begin{tikzpicture}[x=1cm,y=-1cm]
        \board{8}{7}
        \path[use as bounding box] (1,1) rectangle (8,8);
        \foreach \x/\y [count=\i] in {
            1/2, 1/4, 3/3, 2/6, 4/1, 4/5, 5/3, 5/7, 6/1, 6/5, 7/3, 7/7
        }{
            \knight{\x}{\y}{2}
        }

        \draw[line width=0.35mm] (1.83,2.7) -- (3.17,3.3);
        \draw[line width=0.35mm] (1.83,4.3) -- (3.17,3.7);
        \draw[line width=0.35mm] (1.7,4.83) -- (2.3,6.17);
        \draw[line width=0.35mm] (3.7,3.17) -- (4.3,1.83);
        \draw[line width=0.35mm] (3.7,3.83) -- (4.3,5.17);
        \draw[line width=0.35mm] (2.83,6.3) -- (4.17,5.7);

        \draw[line width=0.35mm] (4.7,1.83) -- (5.3,3.17);
        \draw[line width=0.35mm] (4.7,5.17) -- (5.3,3.83);
        \draw[line width=0.35mm] (4.7,5.83) -- (5.3,7.17);
        \draw[line width=0.35mm] (5.7,3.17) -- (6.3,1.83);
        \draw[line width=0.35mm] (5.7,3.83) -- (6.3,5.17);
        \draw[line width=0.35mm] (5.7,7.17) -- (6.3,5.83);

        \draw[line width=0.35mm] (6.7,1.83) -- (7.3,3.17);
        \draw[line width=0.35mm] (6.7,5.17) -- (7.3,3.83);
        \draw[line width=0.35mm] (6.7,5.83) -- (7.3,7.17);

        \foreach \u/\v/\x/\y in {
            1/1/1/1
        }{
            \draw[edge] (\u,\v) -- (\x,\y);
        }
        \foreach \u/\v/\x/\y in {
            1/2/-1/3, 1/4/-1/3, 1/4/-1/5, 7/3/8/1, 7/3/8/5, 7/7/8/5
        }{
            \draw[halfEdge] (\u,\v) -- (\x,\y);
        }
    \end{tikzpicture}
}

\newcommand{\knightWireCrossing}{
    \begin{tikzpicture}[x=1cm,y=-1cm]
        \board{9}{11}
        \foreach \x/\y in {
            1/4, 1/6, 3/5, 3/7, 5/4, 5/6, 7/5, 7/7, 9/4, 9/6, 11/5, 11/7
        }{
            \knightOrange{\x}{\y}{2}
        }
        
        \foreach \x/\y in {
            5/1, 9/1, 3/2, 7/2, 5/3, 9/3, 3/4, 7/4, 5/5, 9/5, 3/6, 7/6, 5/7, 9/7, 3/8, 7/8, 5/9, 9/9
        }{
            \knightBlue{\x}{\y}{2}
        }
        
        \draw[line width=0.45mm, kit-blue] (5.17,1.7) -- (3.83,2.3);
        \draw[line width=0.45mm, kit-blue] (5.83,1.7) -- (7.17,2.3);
        \draw[line width=0.45mm, kit-blue] (9.17,1.7) -- (7.83,2.3);
        \draw[line width=0.45mm, kit-blue] (3.83,2.7) -- (5.17,3.3);
        \draw[line width=0.45mm, kit-blue] (7.17,2.7) -- (5.83,3.3);
        \draw[line width=0.45mm, kit-blue] (7.83,2.7) -- (9.17,3.3);
        
        \draw[line width=0.45mm, kit-blue] (5.17,3.7) -- (3.83,4.3);
        \draw[line width=0.45mm, kit-blue] (5.83,3.7) -- (7.17,4.3);
        \draw[line width=0.45mm, kit-blue] (9.17,3.7) -- (7.83,4.3);
        \draw[line width=0.45mm, kit-blue] (3.83,4.7) -- (5.17,5.3);
        \draw[line width=0.45mm, kit-blue] (7.17,4.7) -- (5.83,5.3);
        \draw[line width=0.45mm, kit-blue] (7.83,4.7) -- (9.17,5.3);
        
        \draw[line width=0.45mm, kit-blue] (5.17,5.7) -- (3.83,6.3);
        \draw[line width=0.45mm, kit-blue] (5.83,5.7) -- (7.17,6.3);
        \draw[line width=0.45mm, kit-blue] (9.17,5.7) -- (7.83,6.3);
        \draw[line width=0.45mm, kit-blue] (3.83,6.7) -- (5.17,7.3);
        \draw[line width=0.45mm, kit-blue] (7.17,6.7) -- (5.83,7.3);
        \draw[line width=0.45mm, kit-blue] (7.83,6.7) -- (9.17,7.3);
        
        \draw[line width=0.45mm, kit-blue] (5.17,7.7) -- (3.83,8.3);
        \draw[line width=0.45mm, kit-blue] (5.83,7.7) -- (7.17,8.3);
        \draw[line width=0.45mm, kit-blue] (9.17,7.7) -- (7.83,8.3);
        \draw[line width=0.45mm, kit-blue] (3.83,8.7) -- (5.17,9.3);
        \draw[line width=0.45mm, kit-blue] (7.17,8.7) -- (5.83,9.3);
        \draw[line width=0.45mm, kit-blue] (7.83,8.7) -- (9.17,9.3);
        
        \draw[line width=0.45mm, dashed, kit-blue] (3.83,0.7) -- (5.17,1.3);
        \draw[line width=0.45mm, dashed, kit-blue] (7.17,0.7) -- (5.83,1.3);
        \draw[line width=0.45mm, dashed, kit-blue] (7.83,0.7) -- (9.17,1.3);
        \draw[line width=0.45mm, dashed, kit-blue] (5.17,9.7) -- (3.83,10.3);
        \draw[line width=0.45mm, dashed, kit-blue] (5.83,9.7) -- (7.17,10.3);
        \draw[line width=0.45mm, dashed, kit-blue] (9.17,9.7) -- (7.83,10.3);

        \draw[line width=0.45mm, kit-orange] (1.83,4.7) -- (3.17,5.3);
        \draw[line width=0.45mm, kit-orange] (1.83,6.3) -- (3.17,5.7);
        \draw[line width=0.45mm, kit-orange] (1.83,6.7) -- (3.17,7.3);
        \draw[line width=0.45mm, kit-orange] (3.83,5.3) -- (5.17,4.7);
        \draw[line width=0.45mm, kit-orange] (3.83,5.7) -- (5.17,6.3);
        \draw[line width=0.45mm, kit-orange] (3.83,7.3) -- (5.17,6.7);
        
        \draw[line width=0.45mm, kit-orange] (5.83,4.7) -- (7.17,5.3);
        \draw[line width=0.45mm, kit-orange] (5.83,6.3) -- (7.17,5.7);
        \draw[line width=0.45mm, kit-orange] (5.83,6.7) -- (7.17,7.3);
        \draw[line width=0.45mm, kit-orange] (7.83,5.3) -- (9.17,4.7);
        \draw[line width=0.45mm, kit-orange] (7.83,5.7) -- (9.17,6.3);
        \draw[line width=0.45mm, kit-orange] (7.83,7.3) -- (9.17,6.7);
        
        \draw[line width=0.45mm, kit-orange] (9.83,4.7) -- (11.17,5.3);
        \draw[line width=0.45mm, kit-orange] (9.83,6.3) -- (11.17,5.7);
        \draw[line width=0.45mm, kit-orange] (9.83,6.7) -- (11.17,7.3);
        
        \draw[line width=0.45mm, dashed, shorten >= 0.7cm, kit-orange] (1.17,4.7) -- (-0.17,5.3);
        \draw[line width=0.45mm, dashed, shorten >= 0.7cm, kit-orange] (1.17,6.3) -- (-0.17,5.7);
        \draw[line width=0.45mm, dashed, shorten >= 0.7cm, kit-orange] (1.17,6.7) -- (-0.17,7.3);
        
        \draw[line width=0.45mm, dashed, shorten >= 0.7cm, kit-orange] (11.83,5.3) -- (13.17,4.7);
        \draw[line width=0.45mm, dashed, shorten >= 0.7cm, kit-orange] (11.83,5.7) -- (13.17,6.3);
        \draw[line width=0.45mm, dashed, shorten >= 0.7cm, kit-orange] (11.83,7.3) -- (13.17,6.7);
    \end{tikzpicture}
}

\newcommand{\knightSetComponent}{
    \begin{tikzpicture}[x=1cm,y=-1cm]
        \board{4}{9}
        \path[use as bounding box] (1,1) rectangle (10,5);
        \foreach \x/\y [count=\i] in {
            1/1, 1/3, 3/2, 3/4, 4/4, 5/1, 5/3, 7/2, 7/4, 9/1, 9/3%
        }{
            \knight[\i]{\x}{\y}{2}
        }
        
        \draw[line width=0.35mm] (1.8,1.65) -- (3.2,2.35);
        \draw[line width=0.35mm] (1.8,3.35) -- (3.2,2.65);
        \draw[line width=0.35mm] (1.8,3.65) -- (3.2,4.35);
        
        \draw[line width=0.35mm] (3.65,2.8) -- (4.35,4.2);
        \draw[line width=0.35mm] (3.8,2.35) -- (5.2,1.65);
        \draw[line width=0.35mm] (3.8,2.65) -- (5.2,3.35);
        \draw[line width=0.35mm] (3.8,4.35) -- (5.2,3.65);
        
        \draw[line width=0.35mm] (5.8,1.65) -- (7.2,2.35);
        \draw[line width=0.35mm] (5.8,3.35) -- (7.2,2.65);
        \draw[line width=0.35mm] (5.8,3.65) -- (7.2,4.35);
        
        \draw[line width=0.35mm] (7.8,2.35) -- (9.2,1.65);
        \draw[line width=0.35mm] (7.8,2.65) -- (9.2,3.35);
        \draw[line width=0.35mm] (7.8,4.35) -- (9.2,3.65);

        \draw[line width=0.35mm, dashed, shorten >= 0.7cm] (1.2,1.65) -- (-0.2,2.35);
        \draw[line width=0.35mm, dashed, shorten >= 0.7cm] (1.2,3.35) -- (-0.2,2.65);
        \draw[line width=0.35mm, dashed, shorten >= 0.7cm] (1.2,3.65) -- (-0.2,4.35);
        
        \draw[line width=0.35mm, dashed, shorten >= 0.7cm] (9.8,1.65) -- (11.2,2.35);
        \draw[line width=0.35mm, dashed, shorten >= 0.7cm] (9.8,3.35) -- (11.2,2.65);
        \draw[line width=0.35mm, dashed, shorten >= 0.7cm] (9.8,3.65) -- (11.2,4.35);
    \end{tikzpicture}
}

\newcommand{\knightSetComponentMoves}{
    \begin{tikzpicture}[x=1cm,y=-1cm]
        \board{4}{9}
        \foreach \x/\y/\c in {
            1/3/0, 3/2/2, 3/4/2, 4/4/2, 5/1/2, 5/3/2, 7/2/2, 7/4/2, 9/1/2, 9/3/2%
        }{
            \knight{\x}{\y}{\c}
        }
        
        \draw [rounded corners = 0.35cm, line width=0.25mm, -latex] (3.2,4.2) -- (1.6,3.5) -- (3.2,2.8);
        \draw [rounded corners = 0.35cm, line width=0.25mm, -latex] (4.2,4.2) -- (3.5,2.5) -- (5.2,3.2);
        \draw [rounded corners = 0.35cm, line width=0.25mm, -latex, dashed] (5.2,1.8) -- (3.7,2.3);
        \node at (4.6,2) {\Huge{\textbf{\red{\texttimes}}}};
    \end{tikzpicture}
}

\newcommand{\knightDecrementComponent}{
    \begin{tikzpicture}[x=1cm,y=-1cm]
        \board{4}{9}
        \foreach \x/\y [count=\i] in {
            1/1, 1/3, 3/2, 3/4, 5/3, 7/2, 7/4, 9/1, 9/3%
        }{
            \knight[\i]{\x}{\y}{2}
        }
        
        \draw[line width=0.35mm] (1.8,1.65) -- (3.2,2.35);
        \draw[line width=0.35mm] (1.8,3.35) -- (3.2,2.65);
        \draw[line width=0.35mm] (1.8,3.65) -- (3.2,4.35);
        
        \draw[line width=0.35mm] (3.8,2.65) -- (5.2,3.35);
        \draw[line width=0.35mm] (3.8,4.35) -- (5.2,3.65);
        
        \draw[line width=0.35mm] (5.8,3.35) -- (7.2,2.65);
        \draw[line width=0.35mm] (5.8,3.65) -- (7.2,4.35);
        
        \draw[line width=0.35mm] (7.8,2.35) -- (9.2,1.65);
        \draw[line width=0.35mm] (7.8,2.65) -- (9.2,3.35);
        \draw[line width=0.35mm] (7.8,4.35) -- (9.2,3.65);

        \draw[line width=0.35mm, dashed, shorten >= 0.7cm] (1.2,1.65) -- (-0.2,2.35);
        \draw[line width=0.35mm, dashed, shorten >= 0.7cm] (1.2,3.35) -- (-0.2,2.65);
        \draw[line width=0.35mm, dashed, shorten >= 0.7cm] (1.2,3.65) -- (-0.2,4.35);
        
        \draw[line width=0.35mm, dashed, shorten >= 0.7cm] (9.8,1.65) -- (11.2,2.35);
        \draw[line width=0.35mm, dashed, shorten >= 0.7cm] (9.8,3.35) -- (11.2,2.65);
        \draw[line width=0.35mm, dashed, shorten >= 0.7cm] (9.8,3.65) -- (11.2,4.35);
    \end{tikzpicture}
}

\newcommand{\knightCheckLabeled}{
    \begin{tikzpicture}[x=1cm,y=-1cm]
        \board{4}{9}
        \path[use as bounding box] (1,1) rectangle (10,5);
        
        \squareColor{5}{1}{blue}
        \foreach \x/\y [count=\i] in {
            1/1, 1/3, 3/2, 3/4, 5/1, 7/2, 9/1%
        }{
            \knight[\i]{\x}{\y}{2}
        }
        
        \draw[line width=0.35mm] (1.8,1.65) -- (3.2,2.35);
        \draw[line width=0.35mm] (1.8,3.35) -- (3.2,2.65);
        \draw[line width=0.35mm] (1.8,3.65) -- (3.2,4.35);
        
        \draw[line width=0.35mm] (3.8,2.35) -- (5.2,1.65);        
        \draw[line width=0.35mm] (5.8,1.65) -- (7.2,2.35);        
        \draw[line width=0.35mm] (7.8,2.35) -- (9.2,1.65);
        
        \draw[line width=0.35mm, dashed, shorten >= 0.7cm] (1.2,1.65) -- (-0.2,2.35);
        \draw[line width=0.35mm, dashed, shorten >= 0.7cm] (1.2,3.35) -- (-0.2,2.65);
        \draw[line width=0.35mm, dashed, shorten >= 0.7cm] (1.2,3.65) -- (-0.2,4.35);
    \end{tikzpicture}
}

\newcommand{\knightThreeVariable}{
    \begin{tikzpicture}[x=1cm,y=-1cm]
        \board{4}{17}
        \squareColor{9}{2}{blue}
        \foreach \x/\y [count=\i] in {
            1/2, 1/4, 3/1, 3/3, 5/2, 7/3, 9/2, 11/3, 13/4, 13/2, 15/1, 15/3, 17/4, 17/2
        }{
            \knight[\i]{\x}{\y}{2}
        }
        
        \draw[line width=0.35mm] (1.8,2.35) -- (3.2,1.65);
        \draw[line width=0.35mm] (1.8,2.65) -- (3.2,3.35);
        \draw[line width=0.35mm] (1.8,4.35) -- (3.2,3.65);
        
        \draw[line width=0.35mm] (3.8,1.65) -- (5.2,2.35);
        \draw[line width=0.35mm] (3.8,3.35) -- (5.2,2.65);
        
        \draw[line width=0.35mm] (5.8,2.65) -- (7.2,3.35);
        \draw[line width=0.35mm] (7.8,3.35) -- (9.2,2.65);
        \draw[line width=0.35mm] (9.8,2.65) -- (11.2,3.35);
        
        \draw[line width=0.35mm] (11.8,3.35) -- (13.2,2.65);
        \draw[line width=0.35mm] (11.8,3.65) -- (13.2,4.35);
        \draw[line width=0.35mm] (13.8,4.35) -- (15.2,3.65);
        
        \draw[line width=0.35mm] (13.8,2.35) -- (15.2,1.65);
        \draw[line width=0.35mm] (13.8,2.65) -- (15.2,3.35);
        
        \draw[line width=0.35mm] (15.8,1.65) -- (17.2,2.35);
        \draw[line width=0.35mm] (15.8,3.35) -- (17.2,2.65);
        \draw[line width=0.35mm] (15.8,3.65) -- (17.2,4.35);

        \draw[line width=0.35mm, dashed, shorten >= 0.7cm] (1.2,2.35) -- (-0.2,1.65);
        \draw[line width=0.35mm, dashed, shorten >= 0.7cm] (1.2,2.65) -- (-0.2,3.35);
        \draw[line width=0.35mm, dashed, shorten >= 0.7cm] (1.2,4.35) -- (-0.2,3.65);
        
        \draw[line width=0.35mm, dashed, shorten >= 0.7cm] (17.8,2.35) -- (19.2,1.65);
        \draw[line width=0.35mm, dashed, shorten >= 0.7cm] (17.8,2.65) -- (19.2,3.35);
        \draw[line width=0.35mm, dashed, shorten >= 0.7cm] (17.8,4.35) -- (19.2,3.65);
    \end{tikzpicture}
}

\newcommand{\knightAddition}{
    \begin{tikzpicture}[x=1cm,y=-1cm]
        \board{9}{11}
        \foreach \i/\x/\y in {
            1/1/4, 2/1/6, 4/3/7, 5/5/4, 6/5/6, 7/7/5, 8/7/7, 9/9/4, 10/9/6, 11/11/5, 12/11/7
        }{
            \knightOrange[\i]{\x}{\y}{2}
        }
        
        \foreach \x/\y [count=\i from 13] in {
            3/1, 5/1, 4/3, 6/3, 5/5
        }{
            \knightBlue[\i]{\x}{\y}{2}
        }
        \knightSplit[3]{3}{5}{2}
        
        \foreach \x/\y [count=\i from 18] in {
            4/7, 5/9, 6/8, 7/8
        }{
            \knight[\i]{\x}{\y}{2}
        }
        
        \draw[line width=0.35mm, kit-orange] (1.8,4.65) -- (3.2,5.35);
        \draw[line width=0.35mm, kit-orange] (1.8,6.35) -- (3.2,5.65);
        \draw[line width=0.35mm, kit-orange] (1.8,6.65) -- (3.2,7.35);
        
        \draw[line width=0.35mm, kit-orange] (3.8,5.35) -- (5.2,4.65);
        \draw[line width=0.35mm, kit-orange] (3.8,5.65) -- (5.2,6.35);
        \draw[line width=0.35mm, kit-orange] (3.8,7.35) -- (5.2,6.65);
        
        \draw[line width=0.35mm, kit-orange] (5.8,4.65) -- (7.2,5.35);
        \draw[line width=0.35mm, kit-orange] (5.8,6.35) -- (7.2,5.65);
        \draw[line width=0.35mm, kit-orange] (5.8,6.65) -- (7.2,7.35);
        
        \draw[line width=0.35mm, kit-orange] (7.8,5.35) -- (9.2,4.65);
        \draw[line width=0.35mm, kit-orange] (7.8,5.65) -- (9.2,6.35);
        \draw[line width=0.35mm, kit-orange] (7.8,7.35) -- (9.2,6.65);
        
        \draw[line width=0.35mm, kit-orange] (9.8,4.65) -- (11.2,5.35);
        \draw[line width=0.35mm, kit-orange] (9.8,6.35) -- (11.2,5.65);
        \draw[line width=0.35mm, kit-orange] (9.8,6.65) -- (11.2,7.35);

        \draw[line width=0.45mm, kit-blue] (3.65,1.8) -- (4.35,3.2);
        \draw[line width=0.45mm, kit-blue] (3.65,5.2) -- (4.35,3.8);
        \draw[line width=0.45mm, kit-blue] (4.65,3.2) -- (5.35,1.8);
        \draw[line width=0.45mm, kit-blue] (4.65,3.8) -- (5.35,5.2);
        \draw[line width=0.45mm, kit-blue] (4.65,7.2) -- (5.35,5.8);

        \draw[line width=0.45mm, kit-blue] (5.65,1.8) -- (6.35,3.2);
        \draw[line width=0.45mm, kit-blue] (5.65,5.2) -- (6.35,3.8);

        \draw[line width=0.35mm] (5.65,6.8) -- (6.35,8.2);
        \draw[line width=0.35mm] (6.65,3.8) -- (7.35,5.2);
        
        \draw[line width=0.35mm] (3.65,5.8) -- (4.35,7.2);
        \draw[line width=0.35mm] (4.65,7.8) -- (5.35,9.2);
        \draw[line width=0.35mm] (4.8,7.65) -- (6.2,8.35);
        \draw[line width=0.35mm] (5.8,9.35) -- (7.2,8.65);

        \draw[line width=0.35mm, dashed, shorten >= 0.7cm, kit-orange] (1.2,4.65) -- (-0.2,5.35);
        \draw[line width=0.35mm, dashed, shorten >= 0.7cm, kit-orange] (1.2,6.35) -- (-0.2,5.65);
        \draw[line width=0.35mm, dashed, shorten >= 0.7cm, kit-orange] (1.2,6.65) -- (-0.2,7.35);
        
        \draw[line width=0.35mm, dashed, shorten >= 0.7cm, kit-orange] (11.8,5.35) -- (13.2,4.65);
        \draw[line width=0.35mm, dashed, shorten >= 0.7cm, kit-orange] (11.8,5.65) -- (13.2,6.35);
        \draw[line width=0.35mm, dashed, shorten >= 0.7cm, kit-orange] (11.8,7.35) -- (13.2,6.65);
    \end{tikzpicture}
}

\usetikzlibrary{calc}

\newcommand{\gadgetBox}[6]{
    \draw (#1 * #6, #2 * #6) rectangle (#1 * #6 + #3 * #6, #2 * #6 + #4 * #6);
    \node[anchor=south west] at (#1 * #6, #2 * #6) {#5};
}

\newcommand{\knightFlipGadget}{
    \begin{tikzpicture}
      \gadgetBox{-2.7}{-1.1}{7.4}{2.8}{Flip}{1}
        
      \node (asym) at (1,-0.3) [draw] {\begin{tabular}{l}
                                    \( (1,1)\ \longleftrightarrow\ (0,0)\) \\
                                    \hline  \( (0,0)\ \longleftrightarrow\ (0,1) \)
                                    \end{tabular}   };
      \node (plusOne) at (-2,1) [circle, draw] { \( \max \) };
      \node (minus) at (0,1) [draw] { \( -(1,0) \) };
      \node (times) at (2,1) [draw] { \( \cdot (1,0) \) };
      \node (plusTwo) at (4,1) [circle, draw] { \( \max \) };

      \node (in) at (-4,1) {$(x,0)$};
      \node (out) at (6,1) {$(0,x)$};

      \draw[-latex] (in) -- (plusOne);
      \draw[-latex] (asym) to[out=180, in=-90] (plusOne);
      \draw[-latex] (plusOne) -- (minus);
      \draw[-latex] (minus) -- (times);
      \draw[-latex] (times) -- (plusTwo);
      \draw[-latex] (asym) to[out=0, in=-90] (plusTwo);      
      \draw[-latex] (plusTwo) -- (out);
    \end{tikzpicture}
}

\newcommand{\knightZeroComponent}{
    \begin{tikzpicture}
      \gadgetBox{-0.9}{-0.5}{5.8}{2}{$\cdot (0,1)$}{1}
      
      \node (plus) at (0,1) [draw] { \( +(1,0) \) };
      \node (minus) at (4,1) [draw] { \( -(1,0) \) };

      \node (in) at (-2,1) {$(x,y)$};
      \node (out) at (6,1) {$(0,y)$};

      \draw[-latex] (in) -- (plus);
      \draw[-latex] (plus) -- (minus) node [midway, fill=white] {$(1,y)$};;
      \draw[-latex] (minus) -- (out);
    \end{tikzpicture}
}

\newcommand{\knightAsymmetricVar}{
    \begin{tikzpicture}
      \gadgetBox{1.1}{-1}{4.7}{2.5}{\begin{tabular}{l}
                                    \( (1,1)\ \longleftrightarrow\ (0,0)\) \\
                                    \hline  \( (0,0)\ \longleftrightarrow\ (1,0) \)
                                    \end{tabular}}{1}
      
      \node (out1) at (0,1) {$(x,x)$};
      \node (out2) at (7,1) {$(1-x,0)$};

      \node (var) at (3,1) [draw] { 3-Val };
      \node (zero) at (5,1) [draw] { \( \cdot(1,0) \) };

      \draw[-latex] (var) -- (out1);
      \draw[-latex] (var) -- (zero);
      \draw[-latex] (zero) -- (out2);
    \end{tikzpicture}
}

\newcommand{\knightVarAbstract}{
    \begin{tikzpicture}
      \gadgetBox{1.1}{-1}{4.7}{2.5}{\begin{tabular}{l}
                                    \( (1,0)\ \longleftrightarrow\ (0,0)\) \\
                                    \hline  \( (0,0)\ \longleftrightarrow\ (1,0) \)
                                    \end{tabular}}{1}
      
      \node (out1) at (0,1) {$(x,0)$};
      \node (out2) at (7,1) {$(1-x,0)$};

      \node (var) at (3,1) [draw] { 3-Val };
      \node (zero) at (5,1) [draw] { \( -(0,1) \) };

      \draw[-latex] (var) -- (out1);
      \draw[-latex] (var) -- (zero);
      \draw[-latex] (zero) -- (out2);
    \end{tikzpicture}
}

\newcommand{\knightAndGadget}{
    \newcommand{\scale}{0.9cm}
    \begin{tikzpicture}[x=\scale,y=\scale]
      \gadgetBox{-4.4}{-0.5}{12.6}{4.1}{AND}{1}
      
      \node (flip) at (-3.5,2.8) [draw] {Flip};
      \node (asym) at (3.4,2.5) [draw] {\begin{tabular}{l}
                                    \( (1,1)\ \longleftrightarrow\ (0,0)\) \\
                                    \hline  \( (0,0)\ \longleftrightarrow\ (1,0) \)
                                    \end{tabular}};
      \node (plusOne) at (-3.5,1) [circle, draw] { \( \max \) };
      \node (plusTwo) at (-0.5,1) [circle, draw] { \( \max \) };
      \node (minusOne) at (2.5,1) [draw] { \( -(1,0) \) };
      \node (minusTwo) at (4.3,1) [draw] { \( -(0,1) \) };
      \node (plusThree) at (7.4,1) [circle, draw] { \( \max \) };

      \node (inOne) at (-5.5,1) {$(x,0)$};
      \node (inTwo) at (-3.5,4.4) {$(y,0)$};
      \node (out) at (9.4,1) {$(x \land y, 0)$};

      \draw[-latex] (inOne) -- (plusOne);
      \draw[-latex] (inTwo) -- (flip);
      \draw[-latex] (flip) -- (plusOne);
      \draw[-latex] (asym) to[out=180, in=90] (plusTwo);
      \draw[-latex] (plusOne) -- (plusTwo)  node [midway, fill=white] {$(x,y)$};
      \draw[-latex] (plusTwo) -- (minusOne) node [midway, fill=white, align=center] {!\\$(1,1)$\\};
      \draw[-latex] (minusOne) -- (minusTwo);
      \draw[-latex] (minusTwo) -- (plusThree) node [midway, fill=white] {$(0,0)$};
      \draw[-latex] (asym) to[out=0, in=90] (plusThree);      
      \draw[-latex] (plusThree) -- (out);
    \end{tikzpicture}
}

\definecolor{kit-green}{RGB}{0, 150, 130}
\colorlet{kit-green100}{kit-green}
\colorlet{kit-green90}{kit-green!90!white}
\colorlet{kit-green80}{kit-green!80!white}
\colorlet{kit-green70}{kit-green!70!white}
\colorlet{kit-green60}{kit-green!60!white}
\colorlet{kit-green50}{kit-green!50!white}
\colorlet{kit-green40}{kit-green!40!white}
\colorlet{kit-green30}{kit-green!30!white}
\colorlet{kit-green25}{kit-green!25!white}
\colorlet{kit-green20}{kit-green!20!white}
\colorlet{kit-green15}{kit-green!15!white}
\colorlet{kit-green10}{kit-green!10!white}
\colorlet{kit-green5}{kit-green!5!white}

\definecolor{kit-blue}{RGB}{70, 100, 170}
\colorlet{kit-blue100}{kit-blue}
\colorlet{kit-blue90}{kit-blue!90!white}
\colorlet{kit-blue80}{kit-blue!80!white}
\colorlet{kit-blue70}{kit-blue!70!white}
\colorlet{kit-blue60}{kit-blue!60!white}
\colorlet{kit-blue50}{kit-blue!50!white}
\colorlet{kit-blue40}{kit-blue!40!white}
\colorlet{kit-blue30}{kit-blue!30!white}
\colorlet{kit-blue25}{kit-blue!25!white}
\colorlet{kit-blue20}{kit-blue!20!white}
\colorlet{kit-blue15}{kit-blue!15!white}
\colorlet{kit-blue10}{kit-blue!10!white}
\colorlet{kit-blue5}{kit-blue!5!white}

\definecolor{kit-royalblue}{RGB}{0, 45, 76}
\colorlet{kit-royalblue100}{kit-royalblue}
\colorlet{kit-royalblue90}{kit-royalblue!90!white}
\colorlet{kit-royalblue80}{kit-royalblue!80!white}
\colorlet{kit-royalblue70}{kit-royalblue!70!white}
\colorlet{kit-royalblue60}{kit-royalblue!60!white}
\colorlet{kit-royalblue50}{kit-royalblue!50!white}
\colorlet{kit-royalblue40}{kit-royalblue!40!white}
\colorlet{kit-royalblue30}{kit-royalblue!30!white}
\colorlet{kit-royalblue25}{kit-royalblue!25!white}
\colorlet{kit-royalblue20}{kit-royalblue!20!white}
\colorlet{kit-royalblue15}{kit-royalblue!15!white}
\colorlet{kit-royalblue10}{kit-royalblue!10!white}
\colorlet{kit-royalblue5}{kit-royalblue!5!white}

\definecolor{kit-iceblue}{RGB}{30, 53, 69}
\colorlet{kit-iceblue100}{kit-iceblue}
\colorlet{kit-iceblue90}{kit-iceblue!90!white}
\colorlet{kit-iceblue80}{kit-iceblue!80!white}
\definecolor{kit-iceblue70}{RGB}{68, 94, 111}
\colorlet{kit-iceblue60}{kit-iceblue70!85!white}
\definecolor{kit-iceblue50}{RGB}{168, 185, 196}
\colorlet{kit-iceblue40}{kit-iceblue50!80!white}
\definecolor{kit-iceblue30}{RGB}{218, 225, 230}
\colorlet{kit-iceblue25}{kit-iceblue30!83!white}
\colorlet{kit-iceblue20}{kit-iceblue30!67!white}
\colorlet{kit-iceblue15}{kit-iceblue30!50!white}
\colorlet{kit-iceblue10}{kit-iceblue30!33!white}
\colorlet{kit-iceblue5}{kit-iceblue30!17!white}

\definecolor{kit-red}{RGB}{162, 34, 35}
\colorlet{kit-red100}{kit-red}
\colorlet{kit-red90}{kit-red!90!white}
\colorlet{kit-red80}{kit-red!80!white}
\colorlet{kit-red70}{kit-red!70!white}
\colorlet{kit-red60}{kit-red!60!white}
\colorlet{kit-red50}{kit-red!50!white}
\colorlet{kit-red40}{kit-red!40!white}
\colorlet{kit-red30}{kit-red!30!white}
\colorlet{kit-red25}{kit-red!25!white}
\colorlet{kit-red20}{kit-red!20!white}
\colorlet{kit-red15}{kit-red!15!white}
\colorlet{kit-red10}{kit-red!10!white}
\colorlet{kit-red5}{kit-red!5!white}

\definecolor{kit-yellow}{RGB}{252, 229, 0}
\colorlet{kit-yellow100}{kit-yellow}
\colorlet{kit-yellow90}{kit-yellow!90!white}
\colorlet{kit-yellow80}{kit-yellow!80!white}
\colorlet{kit-yellow70}{kit-yellow!70!white}
\colorlet{kit-yellow60}{kit-yellow!60!white}
\colorlet{kit-yellow50}{kit-yellow!50!white}
\colorlet{kit-yellow40}{kit-yellow!40!white}
\colorlet{kit-yellow30}{kit-yellow!30!white}
\colorlet{kit-yellow25}{kit-yellow!25!white}
\colorlet{kit-yellow20}{kit-yellow!20!white}
\colorlet{kit-yellow15}{kit-yellow!15!white}
\colorlet{kit-yellow10}{kit-yellow!10!white}
\colorlet{kit-yellow5}{kit-yellow!5!white}

\definecolor{kit-orange}{RGB}{223, 155, 27}
\colorlet{kit-orange100}{kit-orange}
\colorlet{kit-orange90}{kit-orange!90!white}
\colorlet{kit-orange80}{kit-orange!80!white}
\colorlet{kit-orange70}{kit-orange!70!white}
\colorlet{kit-orange60}{kit-orange!60!white}
\colorlet{kit-orange50}{kit-orange!50!white}
\colorlet{kit-orange40}{kit-orange!40!white}
\colorlet{kit-orange30}{kit-orange!30!white}
\colorlet{kit-orange25}{kit-orange!25!white}
\colorlet{kit-orange20}{kit-orange!20!white}
\colorlet{kit-orange15}{kit-orange!15!white}
\colorlet{kit-orange10}{kit-orange!10!white}
\colorlet{kit-orange5}{kit-orange!5!white}

\definecolor{kit-lightgreen}{RGB}{140, 182, 60}
\colorlet{kit-lightgreen100}{kit-lightgreen}
\colorlet{kit-lightgreen90}{kit-lightgreen!90!white}
\colorlet{kit-lightgreen80}{kit-lightgreen!80!white}
\colorlet{kit-lightgreen70}{kit-lightgreen!70!white}
\colorlet{kit-lightgreen60}{kit-lightgreen!60!white}
\colorlet{kit-lightgreen50}{kit-lightgreen!50!white}
\colorlet{kit-lightgreen40}{kit-lightgreen!40!white}
\colorlet{kit-lightgreen30}{kit-lightgreen!30!white}
\colorlet{kit-lightgreen25}{kit-lightgreen!25!white}
\colorlet{kit-lightgreen20}{kit-lightgreen!20!white}
\colorlet{kit-lightgreen15}{kit-lightgreen!15!white}
\colorlet{kit-lightgreen10}{kit-lightgreen!10!white}
\colorlet{kit-lightgreen5}{kit-lightgreen!5!white}

\definecolor{kit-purple}{RGB}{163, 16, 124}
\colorlet{kit-purple100}{kit-purple}
\colorlet{kit-purple90}{kit-purple!90!white}
\colorlet{kit-purple80}{kit-purple!80!white}
\colorlet{kit-purple70}{kit-purple!70!white}
\colorlet{kit-purple60}{kit-purple!60!white}
\colorlet{kit-purple50}{kit-purple!50!white}
\colorlet{kit-purple40}{kit-purple!40!white}
\colorlet{kit-purple30}{kit-purple!30!white}
\colorlet{kit-purple25}{kit-purple!25!white}
\colorlet{kit-purple20}{kit-purple!20!white}
\colorlet{kit-purple15}{kit-purple!15!white}
\colorlet{kit-purple10}{kit-purple!10!white}
\colorlet{kit-purple5}{kit-purple!5!white}

\definecolor{kit-brown}{RGB}{167, 130, 46}
\colorlet{kit-brown100}{kit-brown}
\colorlet{kit-brown90}{kit-brown!90!white}
\colorlet{kit-brown80}{kit-brown!80!white}
\colorlet{kit-brown70}{kit-brown!70!white}
\colorlet{kit-brown60}{kit-brown!60!white}
\colorlet{kit-brown50}{kit-brown!50!white}
\colorlet{kit-brown40}{kit-brown!40!white}
\colorlet{kit-brown30}{kit-brown!30!white}
\colorlet{kit-brown25}{kit-brown!25!white}
\colorlet{kit-brown20}{kit-brown!20!white}
\colorlet{kit-brown15}{kit-brown!15!white}
\colorlet{kit-brown10}{kit-brown!10!white}
\colorlet{kit-brown5}{kit-brown!5!white}

\definecolor{kit-cyan}{RGB}{35, 161, 224}
\colorlet{kit-cyan100}{kit-cyan}
\colorlet{kit-cyan90}{kit-cyan!90!white}
\colorlet{kit-cyan80}{kit-cyan!80!white}
\colorlet{kit-cyan70}{kit-cyan!70!white}
\colorlet{kit-cyan60}{kit-cyan!60!white}
\colorlet{kit-cyan50}{kit-cyan!50!white}
\colorlet{kit-cyan40}{kit-cyan!40!white}
\colorlet{kit-cyan30}{kit-cyan!30!white}
\colorlet{kit-cyan25}{kit-cyan!25!white}
\colorlet{kit-cyan20}{kit-cyan!20!white}
\colorlet{kit-cyan15}{kit-cyan!15!white}
\colorlet{kit-cyan10}{kit-cyan!10!white}
\colorlet{kit-cyan5}{kit-cyan!5!white}

\definecolor{kit-gray}{RGB}{0, 0, 0}
\colorlet{kit-gray100}{kit-gray}
\colorlet{kit-gray90}{kit-gray!90!white}
\colorlet{kit-gray80}{kit-gray!80!white}
\colorlet{kit-gray70}{kit-gray!70!white}
\colorlet{kit-gray60}{kit-gray!60!white}
\colorlet{kit-gray50}{kit-gray!50!white}
\colorlet{kit-gray40}{kit-gray!40!white}
\colorlet{kit-gray30}{kit-gray!30!white}
\colorlet{kit-gray25}{kit-gray!25!white}
\colorlet{kit-gray20}{kit-gray!20!white}
\colorlet{kit-gray15}{kit-gray!15!white}
\colorlet{kit-gray10}{kit-gray!10!white}
\colorlet{kit-gray5}{kit-gray!5!white}

\tikzset{inputsquare/.style={
line width=0.35mm,
-latex,
 rounded corners=0.1cm,
 white,
 fill=kit-cyan60,
}}
\tikzset{outputsquare/.style={
line width=0.35mm,
-latex,
 rounded corners=0.1cm,
 white,
 fill=kit-purple60,
}}

\tikzset{gadgetoverlay/.style={
line width=0.35mm,
-latex,
 rounded corners=0.1cm,
 black,
 fill=kit-gray30,
 opacity=0.7,
}}

\newcounter{pieceID}

\newenvironment{newboard}[2]{
    \setcounter{pieceID}{1}
    \begin{tikzpicture}[x=1cm,y=-1cm]
        \board{#1}{#2}
}{\end{tikzpicture}}

\newcommand{\kingnew}[4][]{
    \coordinate (\Alph{pieceID}) at ($(#2+0.5,#3+0.5)$);
    \ifthenelse{\equal{#4}{x}}{}{
        \node (tmp) at ($(#2+0.5,#3+0.5)$) {\includegraphics[width=0.85cm]{images/king/king_new_#4.pdf}};
        \ifthenelse{\equal{#1}{}}{}{
            \node[inner sep=0cm, outer sep=0cm, fill=white, circle, minimum width=0.3cm, fill opacity=0.9, text opacity=1, yshift=-0.1cm] at (tmp.center) {\Alph{pieceID}};
            \stepcounter{pieceID}
        }
    }
}

\newcommand{\kingList}[4][]{
    \foreach \row [count=\y] in {#4} {
        \foreach \c [count=\x] in \row {
            \kingnew[#1]{\x + #2}{\y + #3}{\c}
        }
    }
}

\newcommand{\kingRectangle}[5][]{
    \foreach \y in {1,...,#5} {
        \foreach \x in {1,...,#4} {
            \kingnew[#1]{\x + #2}{\y + #3}{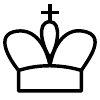}
        }
    }
}

\newcommand{\wireHorizontal}[3][]{
    \foreach \x/\y/\c in {
        2/2/2,
        2/3/2,
        3/2/2,
        3/3/2%
    }{
        \kingnew[#1]{\x + #2}{\y + #3}{\c}
    }
}

\newcommand{\wireConnector}[3][]{
    \foreach \x/\y/\c in {
        1/2/2,
        1/3/2,
        1/4/2,
        2/3/2%
    }{
        \kingnew[#1]{\x + #2}{\y + #3}{\c}}
}

\newcommand{\wireConnectorLeft}[3][]{
    \kingList[#1]{#2}{#3}{
        {x,2},
        {2,2},
        {x,2}%
    }
}

\newcommand{\wireConnectorRight}[3][]{
    \kingList[#1]{#2}{#3}{
        {2,x},
        {2,2},
        {2,x}%
    }
}
\newcommand{\wireConnectorUp}[3][]{
    \kingList[#1]{#2}{#3}{
        {x,2,x},
        {2,2,2}%
    }
}

\newcommand{\wireConnectorDown}[3][]{
    \kingList[#1]{#2}{#3}{
        {2,2,2},
        {x,2,x}%
    }
}

\newcommand{\kingWireCorner}[3][]{
    \foreach \x/\y/\c in {
        1/1/2,
        1/2/2%
    }{
        \kingnew[#1]{\x + #2}{\y + #3}{\c}}
}

\newcommand{\varGadget}[3][]{
    \foreach \x/\y/\c in {
        1/2/2,
        1/3/2,
        2/3/2,
        3/3/2,
        4/3/2,
        5/3/2,
        6/3/2,
        7/2/2,
        7/3/2%
    }{
        \kingnew[#1]{\x + #2}{\y + #3}{\c}
    }
}

\newcommand{\orGadget}[3][]{
    \kingList[#1]{#2}{#3}{
        {x,x,x,x,x},
        {x,x,2,2,x},
        {x,2,2,2,x},
        {x,2,2,2,x},
        {x,x,x,2,x}%
    }
}

\newcommand{\andGadget}[3][]{
    \kingList[#1]{#2}{#3}{
        {x,x,x,x,x,x},
        {x,x,x,2,2,x},
        {x,x,2,2,2,2},
        {x,2,2,x,2,2},
        {x,2,2,2,x,2},
        {x,x,x,2,x,x}%
    }
}

\newcommand{\varGadgetRotated}[3][]{
    \kingList[#1]{#2}{#3}{
        {2,2,2,2,2}%
    }
}

\newcommand{\testGadgetRotated}[3][]{
    \kingList[#1]{#2}{#3}{
        {2,2,2}%
    }
}

\newcommand{\andGadgetRotated}[3][]{
    \kingList[#1]{#2}{#3}{
        {x,2,2,2,x,x},
        {x,x,2,2,2,x},
        {2,2,x,2,2,x},
        {x,2,2,2,x,x},
        {x,2,2,x,x,x},
        {x,x,x,x,x,x}%
    }
}

\newcommand{\orGadgetRotated}[3][]{
    \kingList[#1]{#2}{#3}{
        {x,x,x,2},
        {x,2,2,2},
        {x,2,2,2},
        {x,x,2,2},
        {x,x,x,x}%
    }
}

\newcommand{\gfunc}[3][]{
    \kingList[#1]{#2}{#3}{
        {x,x,2,2,x},
        {x,2,2,2,2},
        {2,2,x,x,2},
        {2,x,x,x,2},
        {x,2,2,2,2},
        {x,x,2,2,x}%
    }
}

\newcommand{\hfunc}[3][]{
    \foreach \x/\y/\c in {
        2/3/2,
        2/4/2,
        2/5/2,
        3/3/2,
        3/4/2,
        4/4/2,
        4/5/2,
        5/2/2,
        5/3/2,
        5/5/2,
        6/2/2,
        6/4/2,
        6/5/2,
        7/3/2%
    }{
        \kingnew[#1]{\x + #2}{\y + #3}{\c}
    }
}

\newcommand{\hfuncrot}[3][]{
    \kingList[#1]{#2}{#3}{
        {x,x,x,x,x},
        {x,x,2,2,2},
        {x,x,2,2,x},
        {x,x,x,2,2},
        {x,2,2,x,2},
        {x,2,x,2,2},
        {x,x,2,x,x}%
    }
}

\newcommand{\kingWireCrossing}{
    \begin{tikzpicture}
      \node (a) at (0,2)  { \( a \) };
      \node (b) at (-2,-2) { \( b \) };
      \node (f) at (0,0) [circle, draw] { CHECK-A };
      \node (g) at (0,-2) [circle, draw] { CHECK-B };
      \node (h) at (0,-4) [circle, draw] { OUT };
      \node (res) at (0,-6)  {$\tilde{a}$};
      \node (b_res) at (4,-2) { \( \tilde{b} \) };

      \node (new_b) at (2,-2) [rectangle, draw] { \( - \ \ \ \tilde{b} \ \ \ + \)};

      \draw [line width=0.25mm, -latex] (a) -- (f);
      \draw [line width=0.25mm, -latex] (b) -- (g);
      \draw [line width=0.25mm, -latex] (new_b.west) -- (f);
      \draw [line width=0.25mm, -latex] (new_b.west) -- (h);
      \draw [line width=0.25mm, -latex] (new_b) -- (b_res);
      \draw [line width=0.25mm, -latex] (f) -- (g);
      \draw [line width=0.25mm, -latex] (g) -- (h);
      \draw [line width=0.25mm, -latex] (h) -- (res);

    \end{tikzpicture}
}

\maketitle

\begin{abstract}
    We extend the study of the 2-\solo problem which was first introduced by Aravind, Misra, and Mittal in 2022.
    2-\solo is a single-player variant of chess in which the player must clear the board via captures such that only one piece remains, with each piece capturing at most twice.
    
    It is known that the problem is solvable in polynomial time for instances containing only pawns, while it becomes \NP-complete for instances restricted to rooks, bishops, or queens.
    In this work, we complete the complexity classification by proving that 2-\solo is \NP-complete if the instance contains only knights or only kings.
\end{abstract}

\section{Introduction}
\label{sec:intro}
\solo is a single-player variant of chess.
Implemented on chess.com \cite{chesscom}, given is a chess position consisting only of pieces of the same color.
The goal is to find a sequence of captures such that only one piece remains, following the rules of chess except that pieces are allowed to capture other pieces of the same color.
Each move has to be a capture, and each piece has a \emph{budget} that limits the number of times it can move. 
In $k$-\solo, each piece can capture at most $k$ times.
The original game on chess.com is implemented with $k = 2$.
(A difference to our variant is that in the chess.com version, the position has at most one king, and if so, it has to be the final remaining piece of the capture sequence.)

The complexity of deciding whether a \solo puzzle can be solved was first studied by Aravind, Misra, and Mittal~\cite{DBLP:journals/tcs/AravindMM24, aravind_et_al:LIPIcs.FUN.2022.5}, where the authors generalize the game to be played on a board of arbitrary size.
They focus on instances that contain only a single piece type and provide complexity results for pawns, queens, rooks, and bishops.
While there is a linear time algorithm for Pawn 2-\solo, they show \NP-completeness for Queen 2-\solo.
Moreover, they prove hardness for Bishop \solo and Rook \solo variants, where each piece receives an individual capture budget, which may be 0, 1, or 2.

The result for rooks was improved by Bilò, Di Donato, Gualà, and Leucci~\cite{DBLP:journals/tcs/BiloDGL25, bilo_et_al:LIPIcs.FUN.2024.4}, who showed that Rook \solo is already \NP-hard if every rook has budget 2. 
With a reduction given in~\cite{aravind_et_al:LIPIcs.FUN.2022.5}, this hardness result immediately extends to 2-\solo with only bishops.
For knights, it was shown in~\cite{DBLP:journals/tcs/BiloDGL25} that 11-\solo is \NP-hard.
For King \solo, no complexity result was known. 

In \cite{unlimited-moves}, a variant is considered where each piece has unlimited budget. 
While all single piece type variants are in \P, combining two piece types makes the problem \NP-hard again.
They also consider a variant where there is a piece (the king in the original game) which must be the final piece remaining.
This problem is again in \P{} for a single piece type, while the picture becomes more diverse when considering the problem where the uncapturable piece is of a different piece type than all the remaining pieces.
Various piece type combinations are discussed, some of which are decidable in polynomial time while others are \NP-hard.

Other variants change the topology of the board such as 1D-\solo (played on a single row) or \textsc{Graph Capture Game} (played on the edges of a given graph)~\cite{aravind_et_al:LIPIcs.FUN.2022.5}.

\subparagraph*{Our Contribution}
In this paper, we complete the complexity classification of 2-\solo with a single piece type.
Based on the first author's Master's thesis~\cite{kuehn2024solochess}, we prove \NP-hardness for both King 2-\solo and Knight 2-\solo by reduction from a well-known variant of 3-SAT.

\subparagraph*{Structure of the Paper} 
Section \ref{sec:prelim} defines some notation and introduces the SAT-variant used in our reductions.
In Section \ref{sec:king} we show \NP-hardness of King 2-\solo, while Knight 2-\solo is considered in Section \ref{sec:knight}.
Proofs omitted in Sections \ref{sec:king} and \ref{sec:knight} are given in Appendix \ref{appendix:king} and \ref{appendix:knight} respectively.
\section{Preliminaries}
\label{sec:prelim}
In this paper, we only consider 2-\solo, and a \emph{piece} is either a king or a knight.
Furthermore, each instance contains only pieces of the same type, i.e., only kings or only knights.
\begin{figure}
    \centering
    \resizebox{0.95\textwidth}{!}{
    \begin{newboard}{3}{5}
        \kingList[id]{0}{0}{
            {x,2,2,x,x},
            {2,2,x,2,x},
            {x,x,2,2,2}%
        }
        \draw[move] (A) -- (B) -- (E);
        \draw[move] (C) -- (D) -- (F);
    \end{newboard}
        \begin{newboard}{3}{5}
            \setcounter{pieceID}{5}
        \kingList[id]{0}{0}{
            {x,x,x,x,x},
            {x,x,x,0,x},
            {x,x,0,2,2}%
        }
        \draw[move] (H) -- (E);
    \end{newboard}
        \begin{newboard}{3}{5}
            \setcounter{pieceID}{5}
        \kingList[id]{0}{0}{
            {x,x,x,x,x},
            {x,x,x,1,x},
            {x,x,0,2,x}%
        }
        \draw[move] (G) -- (F) -- (E);
    \end{newboard}
    }
    \caption{An example instance of King 2-\solo, $2$-kings are in white, $1$-kings in half-red and $0$-kings in red. Middle: Configuration after $\threeSquareMove{A}{B}{E}$ and $\threeSquareMove{C}{D}{F}$. Right: Configuration after $\twoSquareMove{H}{E}$.
    After $\threeSquareMove{G}{F}{E}$, the instance is cleared with $E$ as final square.}
    \label{fig:example-instance}
\end{figure}
A \emph{configuration} is a setup of pieces on a board, together with a \emph{budget} for each piece.
A piece with budget $b \in \set{0, 1, 2}$ is also called a \emph{$b$-piece} (or \emph{$b$-king} or \emph{$b$-knight}).
In each move, a piece with positive budget captures another piece.
The move $\twoSquareMove{X}{Y}$ describes that the piece on square $X$ captures the piece on square $Y$.
If the same piece immediately continues capturing a piece on square $Z$, we use the notation $\threeSquareMove{X}{Y}{Z}$.
Applying a sequence of captures $s$ to a configuration $C$, we obtain a new configuration denoted by $C[s]$.
A capturing sequence $s$ \emph{clears} $C$ if there is only one piece remaining in $C[s]$.
In this case, we call the square of the last piece in $C[s]$ the \emph{final square} of $s$.
Figure~\ref{fig:example-instance} shows a clearable instance of King 2-\solo.

We observe that capturing sequences are monotone regarding budgets.
Let $C$ and $C'$ be configurations with the identical setup of pieces.
We say $C' \geq C$ if the budget of each piece in $C'$ is at least the budget of the corresponding piece in $C$.
\begin{observation}
    \label{obs:monotony}
    Let $C$ be a configuration and $s$ a capturing sequence.
    Let $C'$ be a configuration with $C' \geq C$.
    Then, $s$ is a valid capturing sequence for $C'$, and it is $C'[s] \geq C[s]$.
\end{observation}
For every capturing sequence $s$ of some configuration $C$, there are pieces in $C$ that never move from their original square.
We call these pieces \emph{virtual $0$-pieces} with respect to $s$.
All other pieces are called \emph{virtual $2$-pieces}.
Observe that if we replace each virtual $0$-piece in $C$ with a piece with budget $0$, then $s$ is still a valid capturing sequence.

Before we fix properties on virtual $0$-pieces, it is useful to introduce some graph-theoretical notions.
We define a \emph{capture graph} $G$ of $C$ whose vertex set consists of pieces in $C$.
Two vertices are connected by an (undirected) edge if and only if the pieces can capture each other (disregarding the budget).
For kings and knights, every capture yields a subgraph of the original capture graph.
The \emph{neighborhood} $N(X)$ of a vertex set $X$ is the set of vertices not in $X$ that have an edge to a vertex in $X$.
A \emph{vertex cut} is a set of vertices whose removal disconnects the graph.
The following two lemmas provide properties of virtual $0$-pieces.

\begin{lemma}[restate=zeropieceprops,name=]
    \label{lem:0-piece-properties}
    Let $s$ be a clearing sequence with final square $f$.
    The following properties hold.
    \begin{enumerate}
        \item\label{lem:0-piece-properties:connected} The virtual $0$-pieces form a connected subgraph of the capture graph.
        \item\label{lem:0-piece-properties:vertex-cut} In every vertex cut of the capture graph, there is a virtual $0$-piece.
        \item\label{lem:0-piece-properties:inverse-neighborhood} For every set $X$ of virtual $0$-pieces that does \emph{not} contain $f$, the number of virtual $2$-pieces in the neighborhood of $X$ is at least $|X|$.
    \end{enumerate}
\end{lemma}%
\begin{proof}
    Observe that the final square $f$ holds a virtual $0$-piece (provided the instance contains more than a single piece).
    Any other virtual $0$-piece eventually reaches $f$ through a series of captures, so the path it takes consists entirely of virtual $0$-pieces.
    Thus, within the subgraph of virtual $0$-pieces, each piece is connected to $f$ by a path, which shows the claim.
    
    For the second property, assume that there is a vertex cut where every piece moves.
    But then, the capture graph is disconnected afterwards, and thus, $s$ is not a clearing sequence.

    Finally, let $X$ be a set of virtual $0$-pieces that do not contain $f$.
    Then, for each piece of $X$ its original square is cleared at some point.
    The only way to clear a square without capturing with the piece that is originally present, is to capture into the square with a 2-piece, followed by another capture by that piece.
    As a result, for each virtual $0$-piece of $X$, one neighboring $2$-piece captures it, yielding the claim.
\end{proof}

It is clear that the total number of captured pieces in a capturing sequence does not exceed the sum over all virtual budgets.
More specifically, the budget of the final piece of a clearing sequence is at most the difference between the sum over all virtual budgets and the total number of captured pieces.
We say that a clearing sequence \emph{loses} budget if the budget of the final piece is strictly less than the difference.
In particular, budget is lost if a $1$-piece is captured.
\begin{restatable}{lemma}{losebudget}
    \label{lemma:lose-budget}
    Let $s$ be a clearing sequence that does not lose budget, and let $X$ be a set of virtual $2$-pieces.
    The number of virtual $0$-pieces in the neighborhood of $X$ is at least $|X|$.
\end{restatable}
\begin{proof}
    Let $X$ be a set of virtual $2$-pieces.
    Since $s$ does not lose budget, each piece in $X$ captures a virtual $0$-piece at some point and becomes a $1$-piece, which captures another virtual $0$-piece.
    Thus, no square can be captured by two distinct virtual $2$-pieces, and every virtual $2$-piece has to capture at some point.
    This means that $X$ has at least $|X|$ virtual $0$-pieces in its neighborhood.
\end{proof}

Consider now a capture graph that has a leaf $u$.
Observe that if $u$ is not the final square of some clearing sequence $s$, then the latter contains a capture of $u$ into its only neighbor, and no capture into $u$.
In particular, if $u$ has budget 0, it never leaves its square and we call it \emph{stranded}.
If the neighbor $v$ of $u$ has only one other neighbor $w$, we call $A = (u,v,w)$ a (2-)\emph{antenna} and we can generalize the above observation.
\begin{observation}
    \label{obs:intro:antenna}
    Let $A = (u,v,w)$ be an antenna and let $s$ be a clearing sequence whose final square is not $u$ or $v$.
    Then, $s$ contains the moves $\threeSquareMove{u}{v}{w}$, and it is $C[s] \leq C[s']$ where $s'$ is obtained from $s$ by shifting $\threeSquareMove{u}{v}{w}$ to the beginning of the sequence.
\end{observation}

The definition and observation generalizes to antennae of arbitrary length.
\begin{observation}
\label{obs:intro:long-antenna}
    Let $A = (v_0, \dots, v_b)$ be an antenna in a configuration $C$, and let the leaf $v_0$ of the antenna have a budget of $b'$ with $b' < b$.
    Then there is no capturing sequence that empties each square of $A$.
    Thus, any clearing sequence $s$ of $C$ has its final square in $\{v_0, \dots, v_{b'}\}$, without loss of generality on $v_{b'}$.
\end{observation}

\subsection*{\sat}
In this section we briefly discuss the 3-SAT-variant which we use for our reductions.
A \sat instance consists of a set of variables $U$ and a set of clauses $C$.
Each clause contains two or three literals, and each variable occurs exactly three times in the formula.
Tovey showed that this variant of SAT remains \NP-complete \cite[Thm. 2.1]{TOVEY198485}.
It is easy to show that we may assume that of these three variable occurrences, two are negative and one is positive.

In both reductions, we assume a rectilinear embedding of the \sat instance is given.
The rectilinear embedding is constructed as follows.
There are variables with three outgoing edges and binary OR- and AND-gates, each placed at integer coordinates.
They are connected by rectilinearly drawn directed edges, with turns and crossings at integer coordinates.

\begin{figure}
    \centering
    \includegraphics[page=2,width=0.5\textwidth]{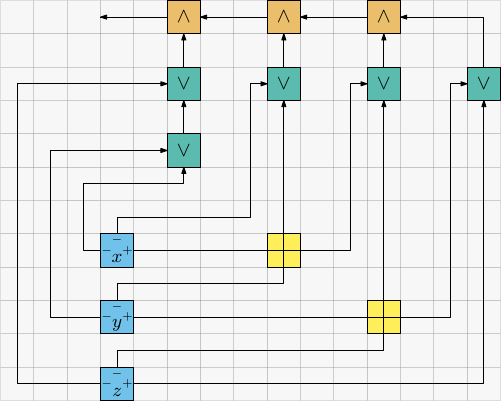}
    \caption{A rectilinear embedding for \sat instance $I = (U, C)$ with $U = \set{x, y, z}$ and $C = \set{(\neg x, \neg y, \neg z), (\neg x, \neg y), (x, \neg z), (y, z)}$. The crossings are indicated by yellow tiles.}
    \label{fig:sat-circuit}
\end{figure}
Then, the embedding is discretized into tiles for the logic gates, variables, straight wires, wire turns and crossings.
See Figure \ref{fig:sat-circuit} for an example embedding.
These tiles fit in a square grid, with variables placed in the first column in the lower rows at sufficient distance to each other.
Clauses are placed to the top right of the variables in distinct columns at sufficient distance.
In particular, the top row contains an AND-tile for each clause except the top-right clause.
Below there are two rows for the up to two OR-tiles of each clause.
Each OR-tile is connected to the tile immediately above it, except the top-right OR-tile, which is connected to the AND-tile to its top left.
AND-tiles are connected to the AND-tile to their left, except for the leftmost AND-tile which outputs the result.
Finally, the first two literals of each clause are connected to the corresponding (lower) OR-tile and, if required, the third literal is connected to the corresponding higher OR-tile.
This placement allows for the edges to be drawn such that they do not go downwards and every pair of edges crosses at most once.
Then it is easy to see that the resulting directed graph with gates, variables and crossings as vertices is acyclic.
The entire embedding is scaled such that each tile has size $\tileSize \times \tileSize$.

\section{King 2-\solo}
\label{sec:king}
We show that King 2-\solo is \NP-hard by reduction from \sat.
Given is a \sat instance with a formula $\Phi$ together with a drawing as described in the previous section.
We now construct an instance $I_\Phi$ of King 2-\solo that mimics the drawing of $\Phi$ such that $\Phi$ is satisfiable if and only if $I_\Phi$ is clearable.
For this, we have two types of \emph{variable assignment} gadgets VAR and 2VAR, AND-, OR-, and crossing gadgets (X), as well as a transportation gadget (so-called \emph{wire}) and a checking gadget (1-TEST).
%
%
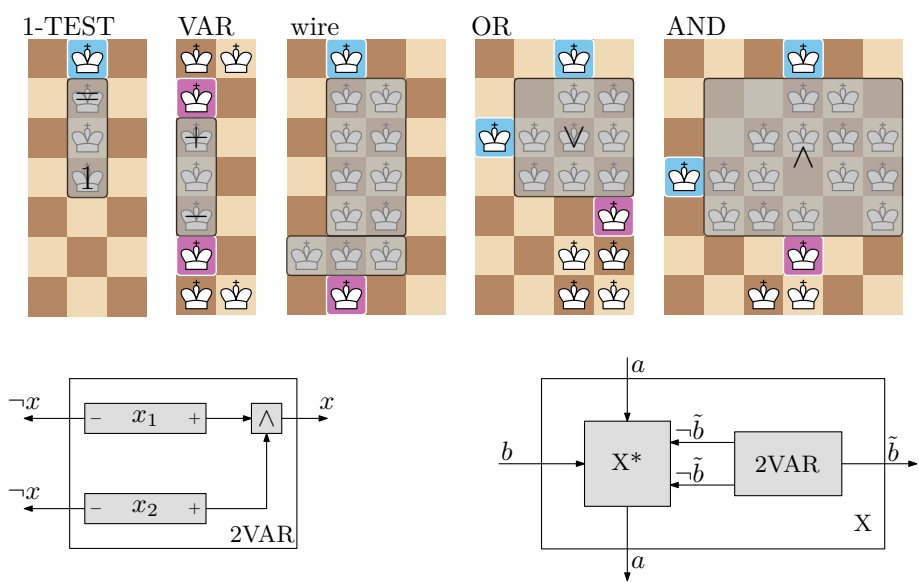
\begin{figure}
    \hspace{1.1cm} 1-TEST \qquad VAR \hspace{0.5cm} wire \hspace{1.5cm} OR \hspace{1.8cm} AND

    \centering
    \resizebox{!}{3.7cm}{
        \begin{newboard}{7}{3}
            \draw[inputsquare] (2,1) rectangle ++(1,1);
            \kingnew[]{2}{1}{2}
            \kingnew[]{2}{2}{2}
            \kingnew[]{2}{3}{2}
            \kingnew[]{2}{4}{2}
            \draw[gadgetoverlay] (2,2) rectangle ++(1,3);
            \node at (2.5,2.5) {\Huge $=$};
            \node at (2.5,4.5) {\Huge $1$};
        \end{newboard}
    }
    \resizebox{!}{3.7cm}{
        \begin{newboard}{7}{2}
            \begin{scope}[xshift=5cm, yshift=0cm, x={(0cm,-1cm)}, y={(-1cm,0cm)}]
            \draw[outputsquare] (2,3) rectangle ++(1,1);
            \draw[outputsquare] (6,3) rectangle ++(1,1);
            \varGadget{0}{0}
        \end{scope}
            \draw[gadgetoverlay] (1,3) rectangle ++(1,3);
            \node at (1.5,3.5) {\Huge $+$};
            \node at (1.5,5.5) {\Huge $-$};
        \end{newboard}
    }
        \resizebox{!}{3.7cm}{
        \begin{newboard}{7}{4}
            \draw[inputsquare] (2,1) rectangle ++(1,1);
            \draw[outputsquare] (2,7) rectangle ++(1,1);
            \kingnew[]{2}{1}{2}
            \wireHorizontal{0}{0}
            \wireHorizontal{0}{2}
            \begin{scope}[xshift=5cm, x={(0cm,-1cm)}, y={(-1cm,0cm)}]
                \wireConnector{5}{-1}
            \end{scope}
            \draw[gadgetoverlay] (2,2) rectangle ++(2,4);
            \draw[gadgetoverlay] (1,6) rectangle ++(3,1);
        \end{newboard}
    }
    \resizebox{!}{3.7cm}{
        \begin{newboard}{7}{4}
            \draw[inputsquare] (1,3) rectangle ++(1,1);
            \draw[inputsquare] (3,1) rectangle ++(1,1);
            \draw[outputsquare] (4,5) rectangle ++(1,1);
            \kingnew[]{1}{3}{2}
            \kingnew[]{3}{1}{2}
            \orGadget{0}{0}
            \kingnew[]{3}{6}{2}
            \kingnew[]{4}{6}{2}
            \kingnew[]{3}{7}{2}
            \kingnew[]{4}{7}{2}
            \draw[gadgetoverlay] (2,2) rectangle ++(3,3);
            \node at (3.5,3.5) {\Huge $\lor$};
        \end{newboard}
    }
    \resizebox{!}{3.7cm}{
        \begin{newboard}{7}{6}
            \draw[inputsquare] (1,4) rectangle ++(1,1);
            \draw[inputsquare] (4,1) rectangle ++(1,1);
            \draw[outputsquare] (4,6) rectangle ++(1,1);
            \kingnew[]{1}{4}{2}
            \kingnew[]{4}{1}{2}
            \andGadget{0}{0}
            \kingnew[]{3}{7}{2}
            \kingnew[]{4}{7}{2}
            \draw[gadgetoverlay] (2,2) rectangle ++(5,4);
            \node at (4.5,4) {\Huge $\land$};
        \end{newboard}
    }

    \centering

    \vspace{0.5cm}
    \includegraphics{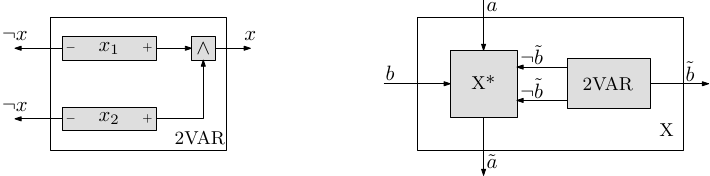}
    \caption{Gadgets from left to right: 1-TEST, VAR, a wire with connector, OR, AND, 2VAR, X. The input squares are marked in blue, the output squares in purple.}
    \label{fig:king:gadgets}
\end{figure}
Each gadget has up to two input and up to two output squares.
All gadgets are depicted in Figure~\ref{fig:king:gadgets}.
We call the wire-, VAR-, OR-, AND-, X* gadget (part of the X gadget), and 1-TEST gadget \emph{atomic gadgets}.
Each non-transportation atomic gadget is placed (possibly rotated or mirrored as needed) in the center of the corresponding tile in the given embedding.
They are connected via wires in a way that adjacent atomic gadgets share a single square, which is an output square in one gadget and an input square in the other.
Other than that, the gadgets do not touch or overlap with other gadgets.
Wires are placed along the edges of the embedding.
Within a tile, wires are extended to connect the tile input and output with the gadget inputs and outputs, respectively.
The size of the tiles for the embedding is chosen such that this is always possible.
The 1-TEST gadget is placed such that its input square is the output square of the wire that corresponds to the circuit output.
See Figure~\ref{fig:king:full-instance} for a full example instance.
\begin{figure}
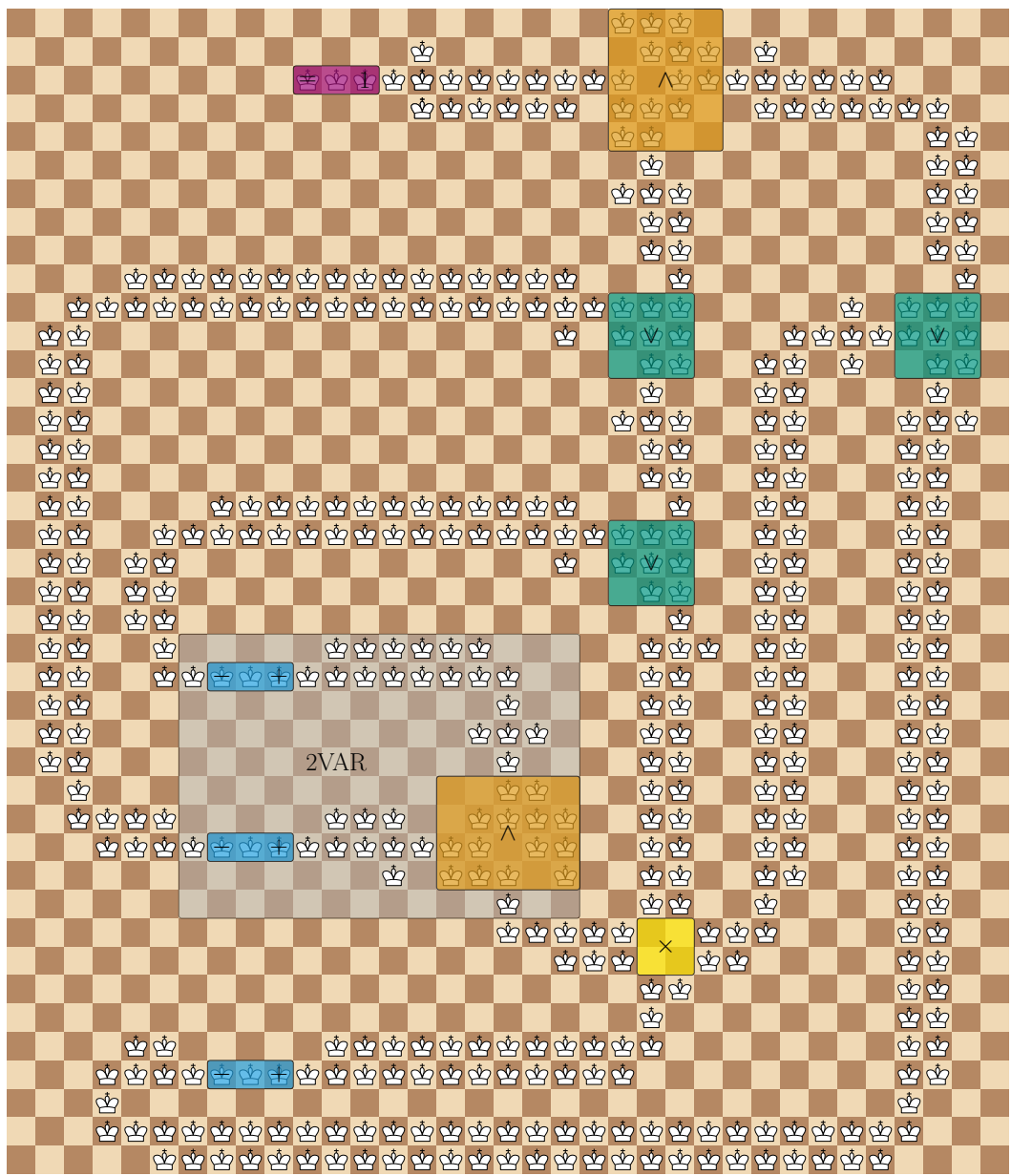

    \centering
    \resizebox{\textwidth}{!}{
        \begin{newboard}{41}{35}

            \draw[gadgetoverlay, opacity=0.5] (7,23) rectangle ++(14,10);
            \node at (12.5, 27.5) {\Huge 2VAR};

            \varGadgetRotated{6}{23}
            \draw[gadgetoverlay, fill=kit-cyan] (8,24) rectangle ++(3,1);
            \node at (8.5,24.5) {\Huge $-$};
            \node at (10.5,24.5) {\Huge $+$};

            \varGadgetRotated{6}{29}
            \draw[gadgetoverlay, fill=kit-cyan] (8,30) rectangle ++(3,1);
            \node at (8.5,30.5) {\Huge $-$};
            \node at (10.5,30.5) {\Huge $+$};

            \andGadget{14}{26}
            \draw[gadgetoverlay, fill=kit-orange] (16,28) rectangle ++(5,4);
            \node at (18.5,30) {\Huge $\land$};

            \wireConnectorRight{13}{28}
            \kingRectangle{11}{28}{2}{2}

            \kingRectangle{11}{22}{6}{2}
            \kingRectangle{17}{23}{1}{2}
            \wireConnectorDown{16}{25}

            \varGadgetRotated{6}{37}
            \draw[gadgetoverlay, fill=kit-cyan] (8,38) rectangle ++(3,1);
            \node at (8.5,38.5) {\Huge $-$};
            \node at (10.5,38.5) {\Huge $+$};

            \orGadgetRotated{20}{17}
            \draw[gadgetoverlay, fill=kit-green] (22,19) rectangle ++(3,3);
            \node at (23.5,20.5) {\Huge $\lor$};

            \orGadgetRotated{20}{9}
            \draw[gadgetoverlay, fill=kit-green] (22,11) rectangle ++(3,3);
            \node at (23.5,12.5) {\Huge $\lor$};

            \orGadgetRotated{30}{9}
            \draw[gadgetoverlay, fill=kit-green] (32,11) rectangle ++(3,3);
            \node at (33.5,12.5) {\Huge $\lor$};

            \andGadgetRotated{20}{0}
            \draw[gadgetoverlay, fill=kit-orange] (22,1) rectangle ++(4,5);
            \node at (24,3.5) {\Huge $\land$};

            \testGadgetRotated{10}{2}
            \draw[gadgetoverlay, fill=kit-purple] (11,3) rectangle ++(3,1);
            \node at (11.5,3.5) {\Huge $=$};
            \node at (13.5,3.5) {\Huge $1$};

            \wireConnectorLeft{13}{1}
            \kingRectangle{14}{2}{6}{2}

            \wireConnectorUp{21}{5}
            \kingRectangle{22}{7}{2}{2}
            \wireConnectorUp{21}{13}
            \kingRectangle{22}{15}{2}{2}

            \kingRectangle{5}{22}{1}{2}
            \kingRectangle{4}{19}{2}{3}
            \kingRectangle{5}{18}{2}{1}
            \kingRectangle{7}{17}{12}{2}
            \wireConnectorRight{19}{17}

            \kingRectangle{3}{28}{3}{2}
            \kingRectangle{2}{27}{1}{2}
            \kingRectangle{1}{11}{2}{16}
            \kingRectangle{2}{10}{2}{1}
            \kingRectangle{4}{9}{16}{2}
            \wireConnectorRight{19}{9}

            \kingRectangle{32}{4}{2}{5}
            \kingRectangle{31}{3}{2}{1}
            \kingRectangle{27}{2}{4}{2}
            \wireConnectorLeft{25}{1}

            \wireConnectorUp{31}{13}
            \kingRectangle{31}{15}{2}{23}
            \kingRectangle{31}{38}{1}{2}
            \kingRectangle{5}{39}{26}{2}
            \kingRectangle{3}{39}{2}{1}
            \kingRectangle{3}{37}{1}{2}
            \kingRectangle{4}{36}{2}{2}

            \kingRectangle{11}{36}{11}{2}
            \kingRectangle{22}{35}{1}{2}
            \kingRectangle{22}{34}{2}{1}
            \kingRectangle{22}{23}{2}{9}
            \wireConnectorUp{22}{21}

            \kingRectangle{17}{32}{2}{1}
            \kingRectangle{19}{32}{3}{2}
            \kingRectangle{24}{32}{2}{2}
            \kingRectangle{26}{31}{1}{2}
            \kingRectangle{26}{12}{2}{19}
            \kingRectangle{27}{11}{2}{1}
            \wireConnectorRight{29}{10}

            \draw[gadgetoverlay, fill=kit-yellow] (23,33) rectangle ++(2,2);
            \node at (24,34) {\Huge $\times$};

        \end{newboard}
    }
    \caption{An example instance $I_\Phi$ for \sat-instance $\Phi$ with variables $\set{x, y}$ and clauses $\set{\set{\neg x, \neg x, y}, \set{x, \neg y}}$.
    The crossing gadget (in yellow) is omitted due to its size.
    Moreover, $I_\Phi$ is constructed without considering tiles to save space.}
    \label{fig:king:full-instance}
\end{figure}


It is useful to assume that in a clearing sequence of $I_\Phi$, the gadgets are cleared in a specific order, starting with the VAR gadgets and following the wires towards the 1-TEST gadget.
We later show that every clearable instance indeed admits such a clearing sequence.
With such an order, we can think of the capturing of kings as a \emph{signal} that goes from each VAR gadget to the 1-TEST gadget.
The budget of the king on an input square of a gadget right before making moves within the gadget indicates an incoming signal.
Analogously, the budget of the king on an output square right after clearing the corresponding gadget indicates an outgoing signal.
We note that the only two possible budgets for both input and output are $0$ and $1$.
These signal values represent false and true, respectively.
In the following, we describe the individual gadgets in more detail.

The 1-TEST gadget checks if its input is $1$.
It is easy to verify that not all kings in the gadget can be captured and thus, it contains the final square in every clearing sequence.
In particular, it is only clearable if its input is a $1$-signal.

The VAR gadget
is connected to a wire on each end, one side representing the positive literal and the other the negative literal.
We note that in every clearing sequence, the middle king of each VAR gadget captures either to the top or to the bottom, producing a $0$-signal in the direction of the capture and a $1$-signal on the other side.
The 2VAR gadget produces two copies of the value on the side representing the positive literal.
As one would expect, variable gadgets never produce $1$-signals on both the positive and the negative side, which we show in the proof of correctness.

The wire gadget propagates signals.
A wire starts with a single input square (the output square of the previous gadget) and ends with a \emph{connector}, which consists of a row with three kings and one row with one king, the output square.
The output square is an input square of the subsequent gadget.
A wire can be placed both vertically and horizontally, and corner gadgets (see Figure~\ref{fig:king:wire-ok}) allow a wire to turn.

The OR and AND gadget both have two input squares and one output square.
The X gadget, which consists of a 2VAR gadget and an X* gadget (omitted due to its size), has two pairs consisting of an input and an output square, and each output square should output the incoming signal at its corresponding input square.

\subparagraph*{Correctness of Reduction}

As the correctness of gadgets requires significant case work, it is convenient to assume for now that the wire, OR, AND, and X gadget work as expected.
This means that they can output the correct value for all possible inputs and cannot do better.
To define this formally, we introduce the notion of \emph{normalized} capturing sequences.
A capturing sequence of a gadget is normalized if it clears the full gadget except for the output squares.
Note that in a gadget with only one output square, every normalized capturing sequence is also a clearing sequence of the gadget (but not vice versa).
Let $f$ be a function that maps $i$ input signals to one output signal.
We say that a gadget with one output \emph{computes} function $f$ if the following holds for every possible combination $X$ of input signals:
\begin{enumerate}
    \item There is a normalized capturing sequence that outputs $f(X)$.
    \item Every normalized capturing sequence outputs $f(X)$ or a signal with lower value.
\end{enumerate}
This notion can be easily extended to the X gadget, which has two outputs.
There should be a normalized capturing sequence such that both outputs are as expected and for every other capturing sequence, no individual output is better than expected.

Before proving that King 2-\solo is \NP-hard, we show that if $I_\Phi$ is clearable, then it can be cleared in a specific order, which we call a \emph{canonical} clearing sequence.
For this, we define the \emph{gadget graph}, whose vertices consist of the atomic gadgets.
There is a directed edge from a gadget $A$ to gadget $B$ if an output square of $A$ is an input square of $B$.
Note that the gadget graph is acyclic and thus admits a topological order of the gadgets where the 1-TEST gadget is last.
In a canonical clearing sequence, the moves are ordered as follows.
\begin{enumerate}
    \item VAR gadgets are cleared first.
    \item If a move is made within an atomic gadget, then all previous atomic gadgets in the gadget graph are already cleared.
    \item The captures within an atomic gadget form a normalized capturing sequence that is consecutive in the clearing sequence.
\end{enumerate}
\begin{restatable}{lemma}{canonicalsequence}
    \label{lemma:canonical-sequence}
    If $I_\Phi$ is clearable, then it admits a canonical clearing sequence.
\end{restatable}

Assuming the correctness of the variable, wire, OR-, AND, and X gadgets, we are now ready to prove that 2-King \solo is \NP-hard.
The idea is that in a canonical clearing sequence, we may determine an input value of a gadget $A$ as the budget of the king on the respective input square right before the first move within $A$ is made.
Thus, following a canonical clearing sequence is equivalent to evaluating the corresponding \sat-instance.
See Appendix~\ref{appendix:correctness} for the proof of the following theorem.
\begin{restatable}{theorem}{kingnphard}
    \label{thm:king:np}
    King 2-\solo is \NP-hard.
\end{restatable}

\subsection*{Correctness of Gadgets}
In the following, we prove that the variable, wire and the OR, AND, and X gadget each compute the desired function.
For every gadget apart from the variable gadgets, we describe a normalized capturing sequence that yields the expected output(s) for every combination of inputs.
Moreover, we show that we cannot do better.
Note that the difference between the total virtual budget and the total number of kings is an upper bound for the output signal.
For every normalized capturing sequence, we show that it has too many virtual $0$-kings or that it loses too much budget to have a better output than expected.

\subparagraph*{Variable Gadgets}
We first argue briefly that the variable gadgets are work as expected.
Note that a VAR gadget never produces a $1$-signal on both the positive and the negative side.
This also holds for 2VAR gadgets: if it produces a $1$-signal on the positive side, then both input values of the AND gadget are $1$.
Thus, both VAR gadgets produce $0$-signals on the negative side.
Moreover, if both a $0$- and a $1$-signal are produced on the side with two outputs, it is never worse to produce two $1$-signals instead by monotony (Observation~\ref{obs:monotony}).
Thus, we may assume that the signals on this side are equal.

\subparagraph*{Wire Gadget}
The wire gadget can be seen in Figure~\ref{fig:king:wire-ok}.
\begin{figure}
        \centering
        \resizebox{0.99\textwidth}{!}{
            \begin{newboard}{3}{5}
                \draw[inputsquare] (1,2) rectangle ++(1,1);
                \draw[outputsquare] (5,2) rectangle ++(1,1);
                \kingnew[id]{1}{2}{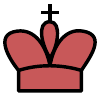}
                \wireHorizontal[id]{0}{-1}
                \wireConnector[id]{3}{-1}

                \draw[move] (B) -- (A) -- (C);
                \draw[move] (D) -- (C) -- (E);
                \draw[move] (F) -- (E) -- (G);
                \draw[move] (H) -- (G) -- (I);
            \end{newboard}
            \begin{newboard}{3}{5}
                \draw[inputsquare] (1,2) rectangle ++(1,1);
                \draw[outputsquare] (5,2) rectangle ++(1,1);
                \kingnew[id]{1}{2}{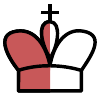}
                \wireHorizontal[id]{0}{-1}
                \wireConnector[id]{3}{-1}
                \draw[move] (A) -- (C);
                \draw[move] (B) -- (C) -- (E);
                \draw[move] (D) -- (E) -- (G);
                \draw[move] (F) -- (G) -- (I);
                \draw[move] (H) -- (I);
            \end{newboard}
            \hspace{1cm}
            \begin{newboard}{4}{5}
                \draw[inputsquare] (1,2) rectangle ++(1,1);
                \draw[outputsquare] (4,3) rectangle ++(1,1);
                \kingnew[id]{1}{2}{0}
                \wireHorizontal[id]{0}{-1}
                \kingnew[id]{4}{2}{2}
                \kingnew[id]{4}{3}{2}
                \kingRectangle[id]{3}{3}{2}{1}
                \draw[move] (B) -- (A) -- (C);
                \draw[move] (D) -- (C) -- (E);
                \draw[move] (F) -- (E) -- (G);
            \end{newboard}
            \begin{newboard}{4}{5}
                \draw[inputsquare] (1,2) rectangle ++(1,1);
                \draw[outputsquare] (4,3) rectangle ++(1,1);
                \kingnew[id]{1}{2}{1}
                \wireHorizontal[id]{0}{-1}
                \kingnew[id]{4}{2}{2}
                \kingnew[id]{4}{3}{2}
                \kingRectangle[id]{3}{3}{2}{1}
                \draw[move] (A) -- (C);
                \draw[move] (B) -- (C) -- (E);
                \draw[move] (D) -- (E) -- (G);
                \draw[move] (F) -- (G);
            \end{newboard}
        }
        \caption{Left: Capturing sequences for a straight wire with input $0$: $\threeSquareMove{B}{A}{C}$, $\threeSquareMove{D}{C}{E}$, $\threeSquareMove{F}{E}{G}$, $\threeSquareMove{H}{G}{I}$, with input $1$: $\twoSquareMove{A}{C}$, $\threeSquareMove{B}{C}{E}$, $\threeSquareMove{D}{E}{G}$, $\threeSquareMove{F}{G}{I}$, $\twoSquareMove{H}{I}$. Right: wire with corner (same capturing sequences as for straight wire).}
        \label{fig:king:wire-ok}
    \end{figure}
We show that the wire can propagate the input signal correctly.
On the other hand, we show that there are no normalized capturing sequences with better output, which yields the following lemma.
\begin{lemma}
    A wire computes the identity function $f(x) = x$ for $x \in \set{0, 1}$.
\end{lemma}
\begin{proof}
    Figure~\ref{fig:king:wire-ok} shows a normalized capturing sequence that produces the correct output for each possible input.
    In both cases, the signal eats up the wire, and every two moves the number of columns with kings is reduced by $1$ while preserving the input signal.

    On the other hand, we show that there is no normalized capturing sequence that outputs $1$ if the input is $0$.
    Let $k$ be the number of columns with two $2$-kings.
    Then, without the input square, the wire has $k + 2$ columns and $2k + 4$ kings in total.
    Each column of the wire except for the input and the output square forms a vertex cut in the capture graph.
    Thus, each of the $k + 1$ columns contains one virtual $0$-king by Lemma~\ref{lem:0-piece-properties:vertex-cut}.
    Moreover, the output king never captures in a normalized clearing sequence and is thus a virtual $0$-king as well.
    Thus, with input $i$, every normalized capturing sequence has at most $i + 2 \cdot (2k + 4 - (k + 1) - 1) = i + 2k + 4$ total virtual budget, and at most $i + 2k + 4$ kings can be cleared.
    Since the input is a $0$-king, which does not contribute additional budget, any clearing sequence results in a $0$-king.
\end{proof}
Normalized capturing sequences for the corner gadget with expected outputs are also depicted in Figure~\ref{fig:king:wire-ok}.
As the capture graph of a wire with corner is a subgraph of the capture graph of a straight wire, there is no normalized capturing sequence for input $0$ that outputs $1$, i.e., a wire with corners also computes the identity function.

\subparagraph*{OR Gadget}
\begin{figure}
        \centering
        \resizebox{!}{0.17\textheight}{
            \begin{newboard}{5}{4}
                \draw[inputsquare] (3,1) rectangle ++(1,1);
                \draw[inputsquare] (1,3) rectangle ++(1,1);
                \draw[outputsquare] (4,5) rectangle ++(1,1);
                \kingnew[id]{3}{1}{0}
                \kingnew[id]{1}{3}{0}
                \orGadget[id]{0}{0}
                \draw[move] (D) -- (A) -- (C);
                \draw[move] (H) -- (B) -- (E);
                \draw[move] (G) -- (C) -- (E);
                \draw[move] (F) -- (E) -- (I);
                \draw[move] (J) -- (I) -- (K);
            \end{newboard}
            \begin{newboard}{5}{4}
                \draw[inputsquare] (3,1) rectangle ++(1,1);
                \draw[inputsquare] (1,3) rectangle ++(1,1);
                \draw[outputsquare] (4,5) rectangle ++(1,1);
                \kingnew[id]{3}{1}{1}
                \kingnew[id]{1}{3}{0}
                \orGadget[id]{0}{0}
                \draw[move] (A) -- (C);
                \draw[move] (H) -- (B) -- (E);
                \draw[move] (D) -- (C) -- (E);
                \draw[move] (F) -- (E) -- (I);
                \draw[move] (G) -- (I) -- (K);
                \draw[move] (J) -- (K);
            \end{newboard}
            \begin{newboard}{5}{4}
                \draw[inputsquare] (3,1) rectangle ++(1,1);
                \draw[inputsquare] (1,3) rectangle ++(1,1);
                \draw[outputsquare] (4,5) rectangle ++(1,1);
                \kingnew[id]{3}{1}{0}
                \kingnew[id]{1}{3}{1}
                \orGadget[id]{0}{0}

                \draw[move] (B) -- (E);
                \draw[move] (D) -- (A) -- (C);
                \draw[move] (G) -- (C) -- (E);
                \draw[move] (H) -- (E) -- (I);
                \draw[move] (F) -- (I) -- (K);
                \draw[move] (J) -- (K);

            \end{newboard}

        }
        \caption{An OR gadget. If both inputs are $0$: $\threeSquareMove{D}{A}{C}, \threeSquareMove{H}{B}{E}, \threeSquareMove{G}{C}{E}, \threeSquareMove{F}{E}{I}, \threeSquareMove{J}{I}{K}$ outputs a $0$-king.
        If input $A$ is $1$, then $\threeSquareMove{H}{B}{E}, \twoSquareMove{A}{C}, \threeSquareMove{D}{C}{E}, \threeSquareMove{F}{E}{I}, \threeSquareMove{G}{I}{K}, \twoSquareMove{J}{K}$ outputs a $1$-king.
        If input $B$ is $1$, then $\threeSquareMove{D}{C}{A}, \twoSquareMove{B}{E}, \threeSquareMove{G}{C}{E}, \threeSquareMove{H}{E}{I}, \threeSquareMove{F}{I}{K}, \twoSquareMove{J}{K}$ outputs a $1$-king.}
        \label{fig:king:or-ok}
    \end{figure}
The OR gadget can be seen in Figure~\ref{fig:king:or-ok} with normalized capturing sequences that produce the correct output for each possible input.
It remains to show that there is no normalized capturing sequence with better output.
\begin{lemma}
    The OR gadget computes the OR-function.
\end{lemma}
\begin{proof}
    We show that if both inputs are $0$, it is impossible to output $1$.
    The capture graph of the OR gadget has three disjoint vertex cuts $\set{C, D}, \set{E, H}, \set{I, J}$.
    By Lemma~\ref{lem:0-piece-properties:vertex-cut}, each vertex cut contains a virtual $0$-king.
    Moreover, the output king never captures in a normalized clearing sequence.
    As the gadget (including output but excluding input) has nine kings in total, every normalized capturing sequence has a total virtual budget of at most $2 \cdot (9 - 4) = 10$ (without the inputs).
    With this budget, at most ten kings can be cleared, which is exactly the number that needs to be captured.
    Since both inputs are $0$-kings, which do not contribute additional budget, every normalized clearing sequence results in a $0$-king.
\end{proof}

\subparagraph*{AND Gadget}
 \begin{figure}
        \centering
        \resizebox{!}{0.17\textheight}{
            \begin{newboard}{6}{6}
                \draw[inputsquare] (4,1) rectangle ++(1,1);
                \draw[inputsquare] (1,4) rectangle ++(1,1);
                \draw[outputsquare] (4,6) rectangle ++(1,1);
                \kingnew[id]{4}{1}{0}
                \kingnew[id]{1}{4}{0}
                \andGadget[id]{0}{0}
                \draw[move] (C) -- (A) -- (D);
                \draw[move] (H) -- (D) -- (G);
                \draw[move] (L) -- (G) -- (K);

                \draw[move] (M) -- (B) -- (I);
                \draw[move] (E) -- (I) -- (J);
                \draw[move] (E) -- (I) -- (J);

                \draw[move] (F) -- (J) -- (O);
                \draw[move] (P) -- (K) -- (O);
                \draw[move] (N) -- (O) -- (Q);
            \end{newboard}
            \begin{newboard}{6}{6}
                \draw[inputsquare] (4,1) rectangle ++(1,1);
                \draw[inputsquare] (1,4) rectangle ++(1,1);
                \draw[outputsquare] (4,6) rectangle ++(1,1);
                \kingnew[id]{4}{1}{1}
                \kingnew[id]{1}{4}{1}
                \andGadget[id]{0}{0}
                \draw[move] (A) -- (C);
                \draw[move] (B) -- (I);
                \draw[move] (M) -- (I) -- (E);
                \draw[move] (F) -- (E) -- (C);
                \draw[move] (D) -- (C) -- (G);
                \draw[move] (H) -- (G) -- (K);
                \draw[move] (L) -- (K);
                \draw[move] (P) -- (K) -- (O);
                \draw[move] (J) -- (O) -- (Q);
                \draw[move] (N) -- (Q);
            \end{newboard}
            \hspace{2cm}
            \begin{newboard}{6}{6}
                \draw[inputsquare] (4,1) rectangle ++(1,1);
                \draw[inputsquare] (1,4) rectangle ++(1,1);
                \draw[outputsquare] (4,6) rectangle ++(1,1);
                \kingnew[id]{4}{1}{0}
                \kingnew[id]{1}{4}{0}
                \andGadget[id]{0}{0}
                \draw[gadgetoverlay] ($(I) + (-0.5, -0.5)$) rectangle ($(M) + (0.5, 0.5)$);
                \draw[gadgetoverlay] ($(C) + (-0.5, -0.5)$) rectangle ($(D) + (0.5, 0.5)$);
                \draw[gadgetoverlay] ($(K) + (-0.5, -0.5)$) rectangle ($(L) + (0.5, 0.5)$);
                \draw[gadgetoverlay] ($(N) + (-0.5, -0.5)$) rectangle ($(O) + (0.5, 0.5)$);
            \end{newboard}

        }
        \caption{An AND gadget.
        Left:
        If both inputs are $0$, then $\threeSquareMove{M}{B}{I}, \threeSquareMove{E}{I}{J}, \threeSquareMove{F}{J}{O}, \threeSquareMove{C}{A}{D}, \threeSquareMove{H}{D}{G}, \threeSquareMove{L}{G}{K}, \threeSquareMove{P}{K}{O}, \threeSquareMove{N}{O}{Q}$ outputs a $0$-king.
        If both inputs are $1$, then $\twoSquareMove{B}{I}, \twoSquareMove{A}{C}, \threeSquareMove{M}{I}{E}, \threeSquareMove{F}{E}{C}, \threeSquareMove{D}{C}{G}, \threeSquareMove{H}{G}{K}, \twoSquareMove{L}{K}, \threeSquareMove{P}{K}{O}, \threeSquareMove{J}{O}{Q}, \twoSquareMove{N}{Q}$ outputs a $1$-king.
        Right: Cut vertex sets which contain one virtual $0$-king each.
        }
        \label{fig:king:and-ok}
    \end{figure}
We show that the AND gadget can output the expected signal, but there is no normalized capturing sequence with better output.
\begin{restatable}{lemma}{kingand}
    The AND gadget computes the AND-function.
\end{restatable}
\begin{proof}
    Figure~\ref{fig:king:and-ok} shows normalized capturing sequences with the desired output if both inputs are $0$ or if both inputs are $1$.
    In the other two cases, where one input is $0$ and the other is $1$, there are normalized capturing sequences that output $0$ by Observation~\ref{obs:monotony}.

    It remains to show that if at least one input is $0$, then it is impossible to output $1$.
    We first argue in the following that the number of virtual $0$-kings is at least $7$ (without the inputs) for every normalized capturing sequence.
    In a normalized capturing sequence, the output square always is a virtual $0$-piece as it does not move.
    The gadget contains the vertex cuts $\set{C, D}, \set{I, M}, \set{K, L}, \set{N, O}$ (see Figure~\ref{fig:king:and-ok}), which contain at least one virtual $0$-king each by Lemma~\ref{lem:0-piece-properties}.\ref{lem:0-piece-properties:vertex-cut}.
    Moreover, virtual $0$-kings are connected by Lemma~\ref{lem:0-piece-properties}.\ref{lem:0-piece-properties:connected}.
    We observe that to connect the virtual $0$-kings in the four vertex cuts and the output king, at least seven virtual $0$-kings in total are needed.
    Without the inputs, this gives us a virtual budget of at most $16$, which is also the number of kings that need to be cleared.
    Thus, if at least one input is $0$, then no capturing sequence with more than seven virtual $0$-kings clears the gadget.

    Consider a normalized capturing sequence $s$ with seven virtual $0$-kings.
    If both inputs are $0$, the output of $s$ cannot be $1$ since the virtual budget is not sufficient.

    If one input is $1$, we argue that $s$ loses budget by showing that there is no placement of the virtual $0$-kings such that all needed properties from Lemma~\ref{lemma:lose-budget} are satisfied.
    We distinguish between cases based on whether king $L$ is a virtual $0$-king.

    If king $L$ is a virtual $0$-king, then the only way to connect the four vertex cuts using seven virtual $0$-kings is to choose $\set{L, G, C, E, I, N}$ and the output king.
    Any other way uses more than seven $0$-virtual kings.
    But then, $\set{H, K, P}$ are three virtual $2$-kings that only have two virtual $0$-kings in their neighborhood, which is a contradiction to Lemma~\ref{lemma:lose-budget}.

    If king $L$ is not a virtual $0$-king, then king $K$ is a virtual $0$-king.
    We further distinguish the cases which kings in $\set{F, G, H}$ are virtual $0$-kings.
    If none of them or only king $F$ is a virtual $0$-king, then $\set{H, L, P}$ are virtual $2$-kings that only have two $0$-kings in their neighborhood.
    If at least two of $\set{F, G, H}$ are virtual $0$-kings, then there are at least eight virtual $0$-kings and the budget is not sufficient to clear the gadget.

    Otherwise, either king $G$ or king $H$ is a virtual $0$-king.
    Then, king $D$ is, too:
    This holds for king $G$ since otherwise, $\set{H, L, P}$ does not have enough virtual $0$-kings in its neighborhood.
    For king $H$, this holds since otherwise, the regions cannot be connected using seven virtual $0$-kings.
    Moreover, the kings $C$, $E$ and $P$ are not virtual $0$-kings since otherwise, the regions cannot be connected using seven virtual $0$-kings.
    With the same argument, king $O$ is a virtual $0$-king.
    Now, we know that $D$, $K$, $O$, $Q$ and either $G$ or $H$ are virtual $0$-kings, while $C$, $E$, $F$, $L$, $P$ and either $G$ or $H$ are not.
    If input $a$ is $1$, then $\set{C, G, H, L, P}$ contains four virtual $2$-kings, but is only connected to at most three virtual $0$-kings ($D$, $K$, and either $G$ or $H$).
    If input $b$ is $1$, then $\set{P, L, H, G, F, E, I, M}$ contains six virtual $2$-kings (since $I$ and $M$ cannot both be chosen).
    However, they are all not connected to input $a$, to king $O$ and the output king.
    Thus, they have at most five virtual $0$-kings in their neighborhood.

    Thus, by Lemma~\ref{lemma:lose-budget}, $s$ loses budget, and the output cannot be $1$.
\end{proof}

\subparagraph*{Crossing Gadget}
\begin{figure}
    \centering
    \includegraphics[page=2]{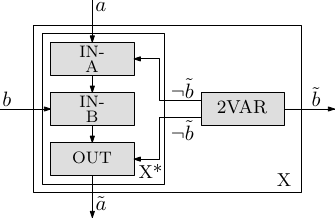}
    \caption{The wire crossing gadget with its four components.}
    \label{fig:king:crossing}
\end{figure}
The crossing gadget takes two inputs $a$ and $b$ and outputs two signals $\tilde{a}$ and $\tilde{b}$.
It contains a 2VAR gadget that feeds $\neg \tilde{b}$ into the X* gadget and outputs $\tilde{b}$.
The idea of the crossing gadget is that there is a normalized capturing sequence if and only if $\tilde{b} \leq b$.
Moreover, in this case, the gadget should output $\tilde{a} \leq a$.
These conditions ensure that an incoming $0$-signal never becomes a $1$-signal.
The first table in Table~\ref{table:crossing} shows the desired outputs for the X* gadget for the three inputs $a$, $b$ and $\neg \tilde{b}$.
\begin{table}
\centering
        \begin{tabular}{ccc   r}
        \toprule
         $a$ & $b$ & $\neg \tilde{b}$ & $X^*(a, b, \neg \tilde{b})$ \\
         \midrule
         0 & 0 & 0 & $\bot$ \\
         0 & 0 & 1 & 0 \\
         0 & 1 & 0 & 0 \\
         0 & 1 & 1 & 0 \\
         1 & 0 & 0 & $\bot$ \\
         1 & 0 & 1 & 1 \\
         1 & 1 & 0 & 1 \\
         1 & 1 & 1 & 1 \\
         \bottomrule
    \end{tabular}
\hfill
    \begin{tabular}{cc   r}
\toprule
$a$ & $\neg \tilde{b}$ & IN-A \\
\midrule
0 & 0 & 0 \\
0 & 1 & 0 \\
1 & 0 & 1 \\
1 & 1 & 2 \\
\bottomrule
\end{tabular}
\hfill
\begin{tabular}{cc   r}
\toprule
IN-A() & $b$ & IN-B \\
\midrule
0 & 0 & 0 \\
0 & 1 & 1 \\
1 & 0 & 0 \\
1 & 1 & 2 \\
2 & 0 & 2 \\
2 & 1 & 2 \\
\bottomrule
\end{tabular}
\hfill
    \begin{tabular}{cc   r}
\toprule
IN-B() & $\neg \tilde{b}$ & OUT \\
\midrule
0 & 0 & $\bot$ \\
0 & 1 & 0 \\
1 & 0 & 0 \\
1 & 1 & 0 \\
2 & 0 & 1 \\
2 & 1 & 1 \\
\bottomrule
\end{tabular}
\caption{Input/output table for the function X*, as well as the components of X*. The $\bot$ symbol indicates that the gadget is not clearable.}
\label{table:crossing}
\end{table}
There are three additional gadgets IN-A, IN-B and OUT that check the inputs (see Figure~\ref{fig:king:crossing}), which are concatenated by a modified wire.
Other than the usual gadgets, the output of IN-A and IN-B as well as the input of IN-B and OUT is not a single square, but a combination of two squares.
Moreover, the wire between IN-A and IN-B as well between IN-B and OUT does not end with a wire connector.
As a consequence, we can input and output more types of signals.
The $0$-signal and the $1$-signal is the same as before: one square has a $0$- or a $1$-king, respectively, and the other square is empty.
Additionally, there is a $2$-signal where one square has a $2$-king and the other square has a $1$-king.
Note that a $2$-signal can be trivially transformed into a $1$-signal.
Moreover, the modified wire can output a $2$-signal if the input is a $2$-signal, see Figure~\ref{king:crossing-gadgets}.
Conversely, we note that $0$- and $1$-inputs can never become $2$-signals.
This means that monotony from Observation~\ref{obs:monotony} also applies to $0$-, $1$-, and $2$-signals.

Table~\ref{table:crossing} shows the functions for IN-A, IN-B, and OUT.
We note that concatenating the three functions as shown in Figure~\ref{fig:king:crossing} yields the desired function X*.
We have a gadget for the IN-B-function and the OUT-function, for the IN-A-function, we can use the same gadget as for IN-B by swapping the inputs.

Assuming the correctness of each gadget, the correctness of the $X$ gadget directly follows.
Figure~\ref{king:crossing-gadgets} shows the IN-B gadget and the OUT gadget.
In the IN-B gadget, the upper input of the IN-B gadget corresponds to the output of the IN-A gadget while the left input is signal~$b$.
In the OUT gadget, the left input corresponds to the output of the IN-B gadget while the upper input is signal $\neg \tilde{b}$.
The following lemmas prove that the gadgets work as expected; see Appendix~\ref{appendix:crossing} for the proofs.
\begin{restatable}{lemma}{kinginb}
    The IN-B gadget computes the IN-B-function, where the upper input signal is in $\set{0, 1, 2}$, and the left input signal is in $\set{0, 1}$.
\end{restatable}
\begin{restatable}{lemma}{kingout}
    The OUT gadget computes the OUT-function, where the upper input signal is in $\set{0, 1}$ and the left input signal is in $\set{0, 1, 2}$.
    There is no normalized capturing sequence if both inputs are $0$.
\end{restatable}

\begin{figure}
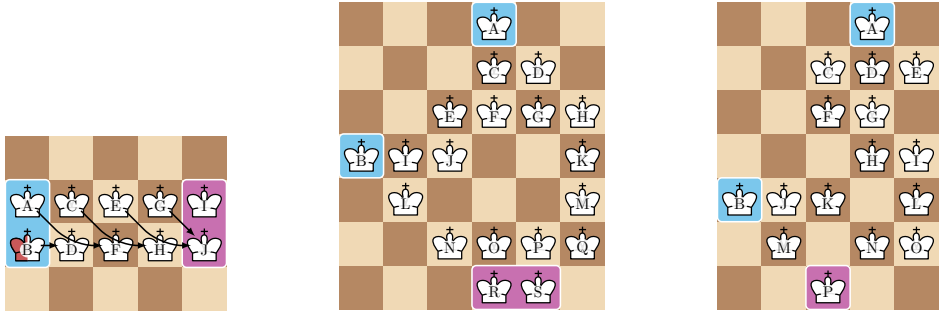

    \centering
    \resizebox{0.9\textwidth}{!}{
        \begin{newboard}{4}{5}
        \draw[inputsquare] (1,2) rectangle ++(1,2);
        \draw[outputsquare] (5,2) rectangle ++(1,2);
        \kingnew[id]{1}{2}{2}
        \kingnew[id]{1}{3}{1}
        \wireHorizontal[id]{0}{0}
        \wireHorizontal[id]{2}{0}
        \draw[move] (B) -- (D);
        \draw[move] (A) -- (D) -- (F);
        \draw[move] (C) -- (F) -- (H);
        \draw[move] (E) -- (H) -- (J);
        \draw[move] (G) -- (J);
    \end{newboard}
    \hspace{2cm}
    \begin{newboard}{7}{6}
        \draw[inputsquare] (1,4) rectangle ++(1,1);
        \draw[inputsquare] (4,1) rectangle ++(2,1);
        \draw[outputsquare] (4,7) rectangle ++(2,1);
        \kingnew[id]{4}{1}{2}
        \kingnew[id]{1}{4}{2}
        \gfunc[id]{1}{1}
        \kingnew[id]{5}{1}{2}
    \end{newboard}
    \hspace{2cm}
    \begin{newboard}{7}{5}
        \draw[inputsquare] (1,5) rectangle ++(1,2);
        \draw[inputsquare] (4,1) rectangle ++(1,1);
        \draw[outputsquare] (3,7) rectangle ++(1,1);
        \kingnew[id]{4}{1}{2}
        \kingnew[id]{1}{5}{2}
        \hfuncrot[id]{0}{0}
        \kingnew[id]{1}{6}{2}
    \end{newboard}
    }
    \caption{Left: Propagation of a $2$-signal in a wire without wire connector. Middle: The IN-B gadget. Right: The OUT gadget.}
    \label{king:crossing-gadgets}
\end{figure}

\section{Knight 2-\solo}
\label{sec:knight}
In this section, we establish the \NP-hardness of Knight 2-\solo, again by reduction from \sat.
The core idea of the reduction is the same as for King 2-\solo, however, the specific gadgets differ.
The main challenge that the knight poses in comparison to the king is the parity constraint in its movement:
a knight on a dark square always moves to a light square and vice versa.
As a result, the capture graph of any knight configuration is bipartite and contains none of the odd cycles that enabled many of the kings gadgets.
We see an example of this in the following subsection on the wire.

Nevertheless, on a high level, the approach is the same:
Given a \sat formula $\Phi$ and a tiled embedding, we construct an instance $I_\Phi$ of Knight 2-\solo that is clearable if and only if $\Phi$ is satisfiable.
To this end, we once more present gadgets for wires, wire crossings, variable assignments, OR-gates, AND-gates, and a \checking.
As before, these gadgets allow for a polynomial-time reduction.
Assuming the correctness of each of the gadgets, we directly conclude:
\begin{theorem}
    Knight 2-\solo is \NP-complete.
\end{theorem}
It remains to discuss the gadgets themselves.
We begin with the wire and the signals that it carries.
\subsection{The Wire}
The wire is shown in Figure \ref{fig:knight:wire}.
\begin{figure}
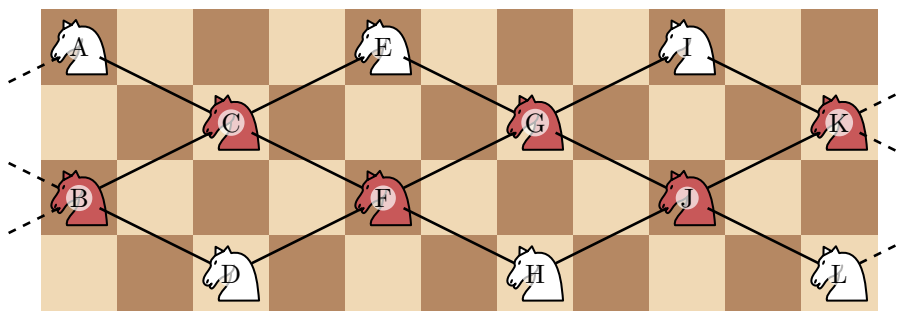
 
    \centering
    \knightWireZeroed
    \caption{The WIRE in Knight 2-\solo. Originally, all pieces have a budget of 2. We show that the knights drawn in red are virtual $0$-pieces.}
    \label{fig:knight:wire}
\end{figure}
Like the king's wire, it consists of columns of two pieces each.
However, for the above mentioned parity reasons, the wire's columns alternate between being placed on light and dark squares, to allow knights to move from one column to the next.
Another difference between the knight and the king coming to bear is that the squares attacked by the former are more spread out.
This sparsity not only ``stretches'' the knight's wire, it also leads to only one of the two diagonal edges between adjacent columns being present.
As a result, one knight of each column has degree 4, while the other knight only has degree 2.

This break of symmetry is reflected in the possible signal values.
Where the king only had one 1-signal, being interpreted as true, the knight's wire can propagate two different 1-signals that are dual to each other.
To account for this, we define our signals as two-dimensional binary vectors.
We encode false as $(0,0)$ and true as $(1,0)$.
These are the signals that travel between variable assignments, OR- and AND-gates, and the final \checking.

Additionally, we give two auxiliary signal values that are (only) used within the AND gadget and allow the latter to be decomposed into simpler parts.
These values are the ``dual 1-signal'' $(0,1)$, as well as the $(1,1)$-signal which can be thought of as a 2-signal.
Once more, these signal values adhere to a form of monotony, where $(0,0) \leq (1,0) \leq (1,1)$ and $(0,0) \leq (0,1) \leq (1,1)$.
The signal values $(1,0)$ and $(0,1)$ are incomparable.

We adapt the notion of a gadget \emph{computing} a function to Knight 2-\solo, using the above defined partial order on signal values.
This would be ill-defined for a gadget which for some input could output either $(1,0)$ or $(0,1)$ (and not $(1,1)$), however, this is never the case in our constructions.
Using this notion we state:
\begin{lemma}[Wire Lemma]
    \label{lem:knight:wire:correctness}
    The wire computes the identity function.
\end{lemma}

In the following, we always assume that the wire is cleared from left to right.
To show the main lemma, we first show that any capturing sequence $s$ clearing the wire indeed induces the highlighted set of virtual $0$-knights.
Note that these are precisely the knights which have degree 4 in the full wire, we call them \emph{inner knights}, while degree two vertices are called \emph{outer knights}.
\begin{lemma}
\label{lem:knight:wire-virtual}
    The set of virtual $0$-knights of a wire is exactly its set of inner knights.
\end{lemma}
\begin{proof}
    Each column of the wire forms a cut, containing one inner knight and one outer knight.
    By Lemma \ref{lem:0-piece-properties}.\ref{lem:0-piece-properties:vertex-cut}, at least one of of these is a virtual $0$-knight.
    Thus, it suffices to show that no outer knight is a virtual $0$-knight.
    Assume for contradiction that there is an outer knight that is a virtual $0$-knight, say $E$.
    By Lemma \ref{lem:0-piece-properties}.\ref{lem:0-piece-properties:inverse-neighborhood}, at least one neighbor of $E$, without loss of generality $G$, is a virtual $2$-knight.
    Then, $F$, $H$ and $J$ are the virtual $0$-knights of the cuts $\{F,G\}$, $\{G,H\}$ and $\{G,J\}$ respectively.
    As a result, both neighbors of virtual $0$-knight $H$ are virtual $0$-knights themselves, a contradiction to Lemma \ref{lem:0-piece-properties}.\ref{lem:0-piece-properties:inverse-neighborhood}.
    The claim follows.
\end{proof}
\begin{corollary}
\label{cor:knight:inner-inner}
    In every capturing sequence that clears the wire from left to right, the inner knight of each column captures the inner knight of the column to its right.
\end{corollary}
\begin{proof} 
    By the previous lemma, the inner knight of the column does not capture the outer knight to its right.
    If it captures to the left, the outer knight in its column becomes a cut vertex and virtual $0$-knight, again violating the previous lemma.
\end{proof}
Using these properties, we define an ordered way to clear the wire:
For each column, all captures from or to its left happen strictly before all captures from or to its right.
Any capturing sequence clearing the wire can be brought into this form by repeatedly swapping consecutive moves that are ordered incorrectly.
If the two moves do not share a square, they can safely be swapped. 
Otherwise, for a move from the right and a move from the left to interfere, they must share their destination square, namely the inner knight of the middle column.
By Corollary \ref{cor:knight:inner-inner}, the move from the right is by an outer knight.
Then, performing the left move \emph{followed by} the right move results in a knight with budget 1, which is no less than without the swap.
Thus, the resulting configuration is greater or equal than without the swap, and by monotony, the capturing sequence remains valid.

Using ordered clearing sequences, we define possible signal values.
Observe that due to the knight's parity constraints, wire columns alternate between \emph{dark-squared} and \emph{light-squared} columns.
We define how to read off the value of a signal on a dark-squared wire column.
This is done right after the final move from the column to its left.
At this point, each column to the left is cleared while each column to the right is untouched.
Our signal then consists of two binary components, so possible signal values are $(0,0), (0,1), (1,0)$ and $(1,1)$.
The first component depends on the outer knight of the signal column.
Since no outer knight is captured, the two options are for it to be present as a 2-knight, or to not be present at all.
In the first case the component has value $1$, in the second case it has value $0$.
The second component depends on the inner knight of the column.
Being a virtual $0$-knight, it is not present as a $2$-knight, only with budget $0$ or $1$ (the latter occurs if it was captured by the outer knight to its left).
We define the second signal component as the budget of this inner knight.
Figure \ref{fig:knight:wire:signals} shows two possible signal values and how to propagate them through the wire.
\begin{figure}
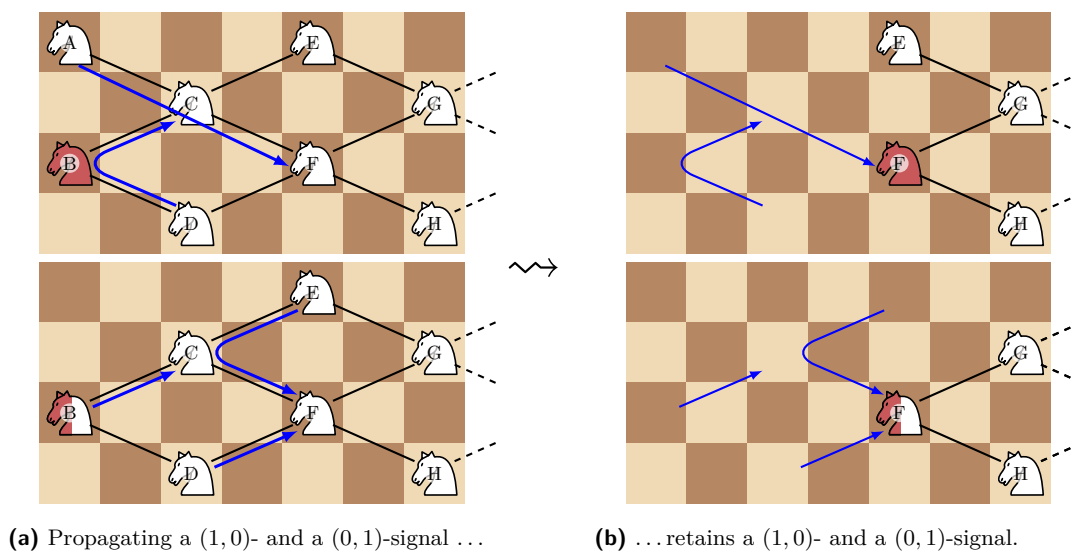

    \begin{subfigure}{0.45\textwidth}
        \centering
        \scalebox{0.8}{
            \knightSignals
        }
        \caption{Propagating a $(1,0)$- and a $(0,1)$-signal \dots}
    \end{subfigure}
    \begin{subfigure}{0.09\textwidth}
        \centering
        \huge{$\rightsquigarrow$}
        \vspace{3.6cm}
    \end{subfigure}
    \hfill
    \begin{subfigure}{0.45\textwidth}
        \centering
        \scalebox{0.8}{
            \knightAfterSignals
        }
        \caption{\dots retains a $(1,0)$- and a $(0,1)$-signal.}
    \end{subfigure}
    
    \caption{Two Knight Wires propagating a signal each.}
    \label{fig:knight:wire:signals}
\end{figure}
The $(1,0)$- and $(0,1)$-signals are unusual in that they are dual to each other.
A $(1,0)$-signal reads off a dark-squared column the same as a $(0,1)$-signal would on a light-squared column.
Nevertheless, for parity reasons a $(1,0)$-signal is never converted into a $(0,1)$-signal and vice versa.
We see this in the proof of Lemma \ref{lem:knight:wire:correctness}.
\begin{proof}[Proof of the \protect{\hyperref[lem:knight:wire:correctness]{Wire Lemma}}]
    Figure \ref{fig:knight:wire:signals} sketches how to obtain capturing sequences for each of the four signal values.
    In particular, for a second component of $0$ or $1$, the sequence starts with captures $\threeSquareMove{D}{B}{C}$ or $\twoSquareMove{B}{C}$ respectively.
    For a first component of $0$ or $1$ it continues with $\threeSquareMove{E}{C}{F}$ or $\threeSquareMove{A}{C}{F}$.
    Finally, for a second component of $1$, a final move $\twoSquareMove{D}{F}$ is appended.

    To see that signal values never increase, we again consider the difference between the total virtual budget and the number of pieces of the wire.
    For the extreme case of just one column, this difference is precisely one lower than the number of $1$'s in the signal, i.e., for signal $(0,0)$ it is -1, for signal $(1,1)$ it is 1 and for signals $(0,1)$ and $(1,0)$ it is 0.
    Each additional wire column contributes two units of virtual budget and two pieces.
    Thus, the difference never increases.
    It follows that signal values propagated through the wire don't increase.

    Finally, we show that a $(1,0)$-signal indeed is never converted into a $(0,1)$-signal and vice versa.
    By the non-increasing difference above, any capturing sequence converting one into the other does not lose budget.
    Consider the configuration of a $(1,0)$-signal, as seen in Figure \ref{fig:knight:wire:signals}.
    The signal-column $0$-knight leaves its square eventually by a double capture from an adjacent $2$-knight, $\threeSquareMove{D}{B}{C}$.
    As this move does not interfere with any other moves, we may assume that it is the first pair of moves in the capturing sequence.
    In the resulting configuration, knight $C$ has two neighbors with a budget of $2$, $A$ and $E$.
    If it is captured by $E$, the eventual capture of $A$ to its only neighbor loses budget.
    Otherwise, $A$ captures $C$ and $F$, restoring a $(1,0)$-signal.
    The reverse case is analogous.
\end{proof}

To align wires and gadgets on the board, we give a number of placements of the wire on the chess board.
We note that each of these correspond to the same capture graph, thus, the wire functions identically in each of the placements.
Figure \ref{fig:knight:wire:corner-thick} shows a wire turning a corner, as well as a wire turning into a \emph{thick wire}.
\begin{figure}
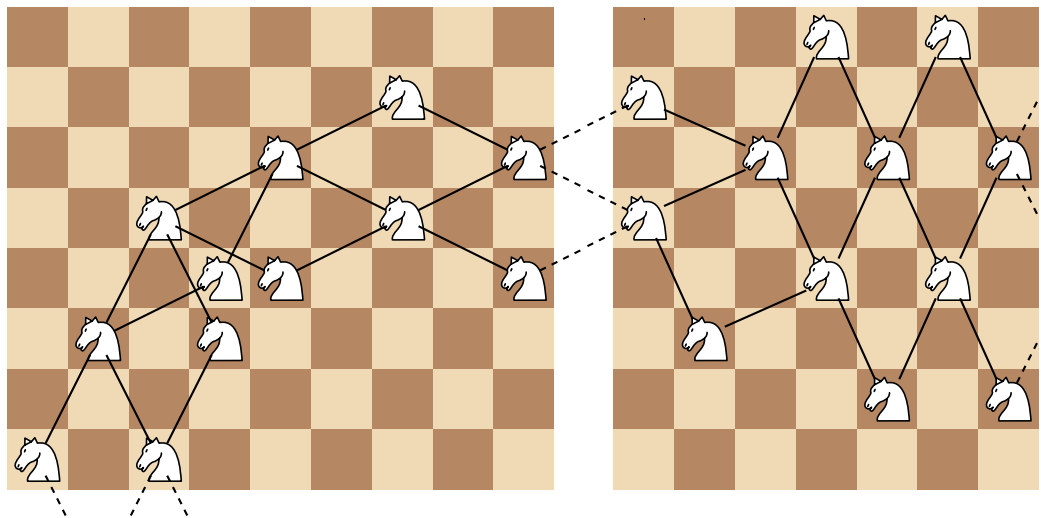

    \centering
    \begin{subfigure}[t]{0.55\textwidth}
        \vspace{0pt}
        \scalebox{0.8}{
            \knightWireCorner
        }
    \end{subfigure}
    \hfill
    \begin{subfigure}[t]{0.43\textwidth}
        \vspace{0pt}
        \scalebox{0.8}{
            \knightHighWire
        }
    \end{subfigure}
    \caption{Two WIRE placements. Left: A wire turning a corner. Right: A standard wire turning into a thick wire.}
    \label{fig:knight:wire:corner-thick}
\end{figure}
Combining these enables the \emph{wire crossing}, seen in Figure \ref{fig:knight:wire:crossing}, without the need for a dedicated wire crossing gadget.
\begin{figure}
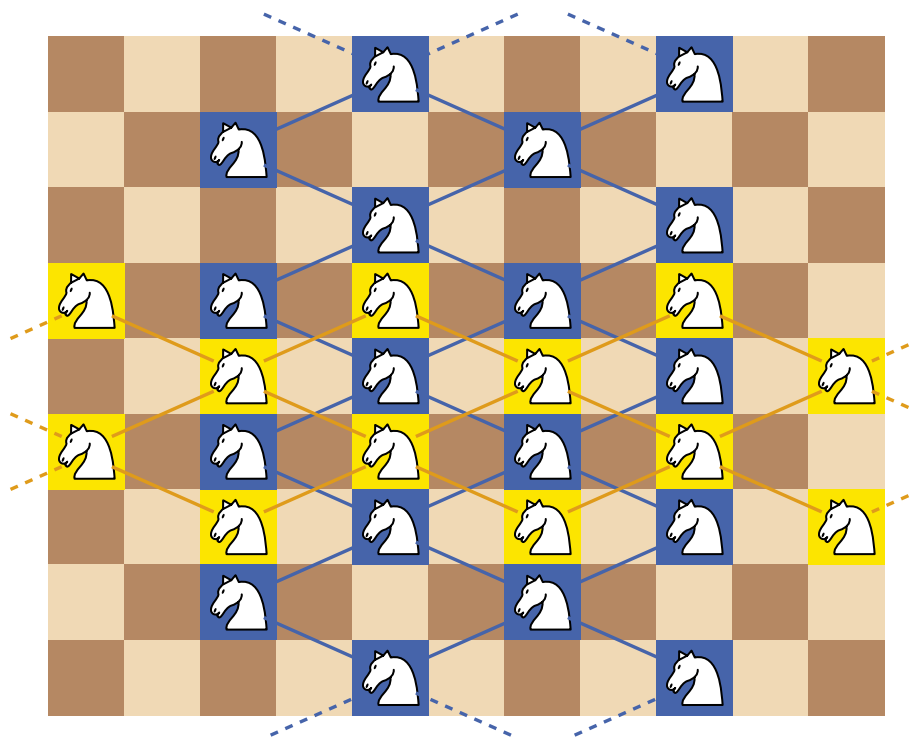

    \centering
    \knightWireCrossing
    \caption{A wire crossing of a yellow and a blue wire utilizing the \emph{thick wire} placement. Note that no yellow piece can capture any blue pieces and vice versa.}
    \label{fig:knight:wire:crossing}
\end{figure}

While the king wire can have any length, the knight wire only repeats every four columns.
In Figure \ref{fig:appendix:knight:wire-shift}, we give additional placements to shift the wire (diagonally) by a single square to be able to align wires correctly for their gadgets.
\begin{figure}
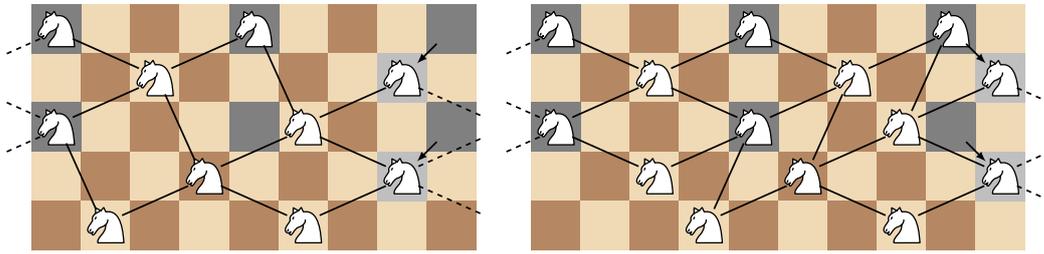

    \scalebox{0.65}{\knightWireDiagonalBackward}
    \scalebox{0.65}{\knightWireDiagonalForward}
    \caption{Shifting a wire diagonally backward and forward respectively.}
    \label{fig:appendix:knight:wire-shift}
\end{figure}
Note that due to parity, dark-squared columns are always mapped to dark-squared columns, so it is impossible to shift the wire by a single non-diagonal square.
However, as inputs to gadgets are of a fixed color, they can still always be aligned correctly.

Finally, after passing a series of corners and gadgets, a wire may be oriented so as to have the dark-squared column be aligned towards the top or towards the bottom.
To switch between those two modes, we give a rotation placement, seen in Figure \ref{fig:appendix:knight:wire-rotate}.
\begin{figure}
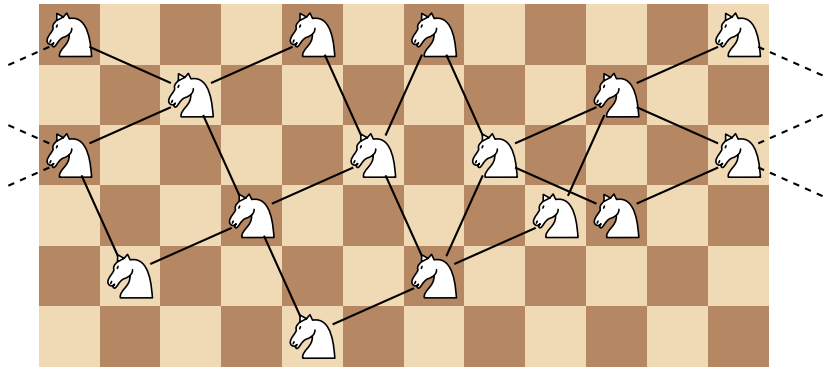

    \centering
    \scalebox{0.8}{
        \knightRotatedWire
    }
    \caption{A rotation placement yields a wire with light-squared columns on top and dark-squared columns below.}
    \label{fig:appendix:knight:wire-rotate}
\end{figure}
%

\subsection{The \checking Gadget}
The \checking gadget for the knight not only functions the same but is also implemented analogously to its king counterpart.
Shown in Figure \ref{fig:knight:check}, it consists of an input wire connected to a path of length 3.
\begin{figure}
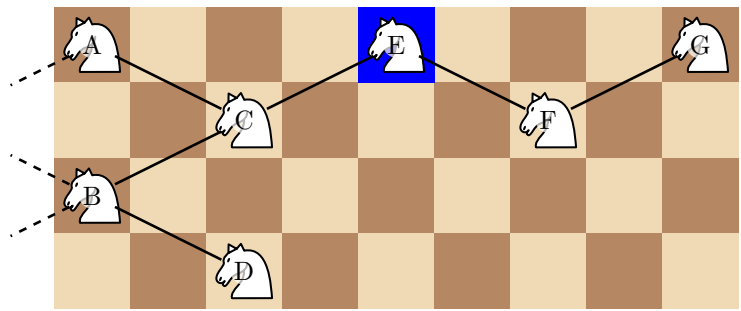

    \centering
    \knightCheckLabeled
    \caption{The \checking gadget for the knight contains the final square for each clearing sequence.}
    \label{fig:knight:check}
\end{figure}
This antenna of length 3 guarantees, by \cref{obs:intro:long-antenna}, that any clearing sequence of an instance has its final square in the \checking gadget, without loss of generality on $E$.
\begin{lemma}
\label{lem:knight:check}
    The \checking gadget reduces to a single piece if it receives a $(1,0)$-signal as input.
    If it receives a $(0,0)$-signal as input, it does not reduce to a single piece.
\end{lemma}
\begin{proof} 
    For input $(1,0)$, knight $A$ is present with a budget of 2.
    Then the sequence $\threeSquareMove{D}{B}{C}, \threeSquareMove{A}{C}{E}, \threeSquareMove{G}{F}{E}$ reduces the gadget to a single piece, as claimed.

    For input $(0,0)$, knight $A$ is not present.
    By Lemma \ref{obs:intro:antenna}, any clearing sequence $s$ (with final square $E$) can be assumed to start with captures $\threeSquareMove{G}{F}{E}$.
    Since knight $B$ is the only neighbor of $D$, it is captured by $D$ at some point.
    Its only path to final square $E$ is via $C$, so $s$ contains the capture $\twoSquareMove{B}{C}$.
    As knight $B$ has a budget of at most 1 before this capture, it results in a 0-leaf on $C$, only adjacent to 0-leaf $E$.
    We conclude that no such clearing sequence $s$ exists.
\end{proof}

\subsection{The Variable Assignment}
%
%
We now review the VAR gadget.
It assigns the value of the literals of a binary variable $x$, having outputs $x$ and $\lnot x$, by producing a true signal to the left and a false signal to the right, or a false signal to the left and a true signal to the right.
Translated to our two-dimensional signals, where true is encoded by $(1,0)$ and false is encoded by $(0,0)$ we state:
\begin{lemma}
\label{lem:knight:var-assign}
    A variable assignment gadget produces one $(1,0)$- and one $(0,0)$-signal.
\end{lemma}
While an analogue to the king VAR gadget would be valid, we provide a slightly more general construction that proves useful for the remaining gadgets.
To this end, we decompose the gadget into two parts: the \emph{3-valued production} gadget (3-VAL), and the Decrement gadget.

The 3-VAL gadget, shown in Figure \ref{fig:knight:var-assign:3-var}, has no inputs and two outputs.
\begin{figure}
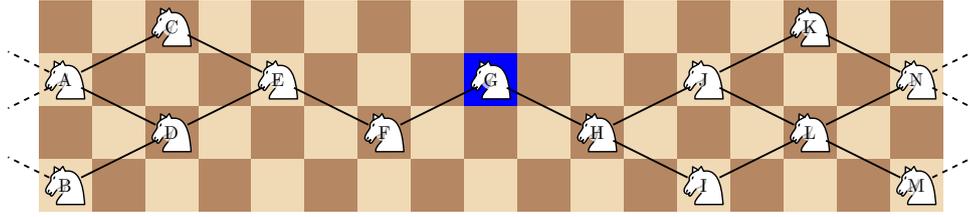

    \centering
    \scalebox{0.7}{\knightThreeVariable}
    \caption{The three valued production gadget produces left and right outputs $(1,1) \leftrightarrow (0,0)$ or $(1,0) \leftrightarrow (0,1)$ or $(0,0) \leftrightarrow (1,1)$.}
    \label{fig:knight:var-assign:3-var}
\end{figure}
It can produce three different output pairs.
\begin{lemma}[3-VAL Lemma]
\label{lem:knight:3-val}
    The 3-VAL gadget produces either a $(1,1)$-signal to the left and a $(0,0)$-signal to the right, or a $(1,0)$-signal to the left and a $(0,1)$-signal to the right, or a $(0,0)$-signal to the left and a $(1,1)$-signal to the right.
\end{lemma}
We combine it with the Decrement gadget (\cref{fig:knight:decrement}) which as the name suggests, decrements a component of a signal.
\begin{figure}
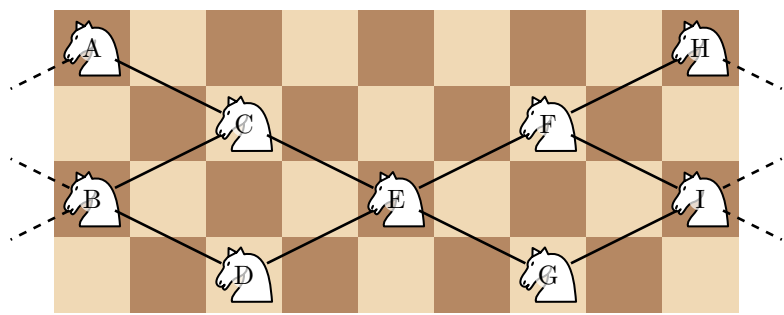

    \centering
    \knightDecrementComponent
    \caption{The first-component Decrement gadget is a wire with an outer knight missing.}
    \label{fig:knight:decrement}
\end{figure}
\begin{lemma}[Decrement Lemma]
\label{lem:knight:decrement}
    The $-(1,0)$ gadget computes the function that maps $(1,y)$ to $(0,y)$ and $(0,y)$ to $\bot$.
    An analogous $-(0,1)$ gadget computes the function that maps $(x,1)$ to $(x,0)$ and $(x,0)$ to $\bot$.
\end{lemma}

The resulting VAR gadget is shown in \cref{fig:knight:var-abstract}.
\begin{figure}
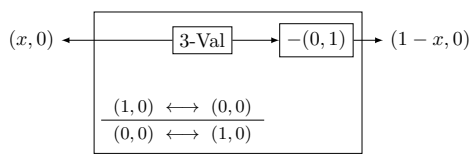

    \centering
    \scalebox{0.75}{
        \knightVarAbstract
    }
    \caption{Using a 3-VAL and a Decrement gadget to build a VAR gadget.}
    \label{fig:knight:var-abstract}
\end{figure}
For a 3-VAL output of $(1,1) \leftrightarrow (0,0)$ this produces an error, the other two possible outputs yield the expected outputs of the VAR gadget.
This shows \cref{lem:knight:var-assign}.

We now prove the correctness of the two building blocks.
\begin{proof}[Proof of the \protect{\hyperref[lem:knight:3-val]{3-VAL Lemma}}]
    Consider knight $G$, highlighted in blue.
    Observe that, by \cref{obs:intro:long-antenna}, no clearing sequence contains the capture $\twoSquareMove{H}{G}$.
    Instead, knight $G$ can be captured via $\twoSquareMove{F}{G}$ which after resolving captures $\twoSquareMove{G}{H}, \twoSquareMove{E}{D}$ yields a $(1,1)$-signal on the left and a $(0,0)$-signal on the right wire, the first of the three possible outcomes.
    
    The remaining options are those where the blue knight is not captured and instead makes a capture itself:
    If it captures to the left through $\threeSquareMove{G}{F}{E}, \threeSquareMove{C}{E}{D}$, this creates a $(0,0)$-signal on the left and (after captures $\threeSquareMove{H}{J}{L}, \twoSquareMove{I}{L}$) a $(1,1)$-signal on the right wire.
    If instead it captures to the right through $\twoSquareMove{G}{H}$, this creates a $(0,1)$-signal on the right.
    Then, after a captures $\threeSquareMove{F}{E}{D}$, the left wire receives a $(1,0)$-signal.
    This covers all possible cases for the blue knight and the capture sequence as a whole, and yields the claimed three possible outcomes.
\end{proof}
\begin{proof}[Proof of the \protect{\hyperref[lem:knight:decrement]{Decrement Lemma}}]
    We first consider the case where the first input component is 0, i.e., knight $A$ is not present.
    Then both neighbors $B$ and $E$ of virtual $0$-knight $C$ are themselves virtual $0$-knights, so by Lemma \ref{lem:0-piece-properties}.\ref{lem:0-piece-properties:inverse-neighborhood}, the gadget cannot be fully cleared.

    Next consider the case where the first input component is 1, i.e., knight $A$ is present.
    Propagating the signal as usual in a wire yields the claimed output.
    We show that no better output is possible:
    The first output component is 0 by the fact that the would-be outer knight of the column of $E$ already is not present.
    Furthermore, a virtual budget argument shows that any signal is decreased by at least 1.
    As a result, $(1,0)$ is mapped to (at most) $(0,0)$ and $(1,1)$ is mapped to (at most) $(0,1)$.
\end{proof}
In the following subsections we utilize the counterpart of the Decrement gadget, namely the Increment gadget.
\begin{figure}
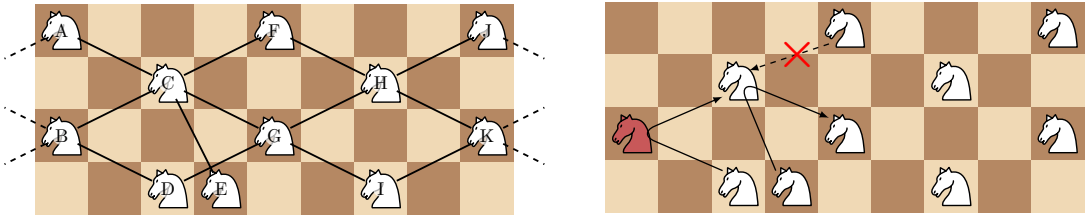

    \centering
    \scalebox{0.7}{\knightSetComponent}
    \hfill
    \scalebox{0.7}{\knightSetComponentMoves}
    \caption{Left: The first-component Increment gadget is a wire with an extra leaf.\\
    Right: Example sequence that increments the first component of a signal.}
    \label{fig:knight:increment}
\end{figure}
The first-component Increment (or $+(1,0)$) gadget increments the first component of a signal, as shown in Figure \ref{fig:knight:increment}.
\begin{lemma}
    The $+(1,0)$ gadget computes the function $(x,y) \mapsto (1,y)$ for $x,y \in \{0,1\}$.
    An analogous $+(0,1)$ gadget computes the function $(x,y) \mapsto (x,1)$ for $x,y \in \{0,1\}$.
\end{lemma}

\begin{proof}
    Figure \ref{fig:knight:increment} shows that there exists a capture sequence that sets the first component to 1.
    To see that no larger (or incomparable) output is possible, it suffices to show that if the second component of the output is 1, then so was the second component of the input.
    For this, consider an output of $(x,1)$, read off in the column containing knights $F$ and $G$, i.e., knight $G$ having a budget of 1 and all knights to its right being untouched.
    We now trace back how the knight got there.
    It was captured by a 2-knight, which was knight $D$, as $C$ is a virtual 0-knight and $I$ has not yet moved.
    Then $D$ did not capture $B$ and, thus, to leave its square towards $C$, knight $B$ was a 1-knight, meaning the input had a second component of 1.

    The $+(0,1)$ gadget has a leaf adjacent to a dark-squared inner knight and functions identically.
\end{proof}
This concludes our discussion of the VAR gadget.
We remark that as in the king's case, we can combine two VAR gadgets and an AND gadget to create the 2VAR gadget (Figure \ref{fig:king:gadgets}).

\subsection{The OR Gadget}
Observe that the OR function can be defined as the maximum function on binary inputs.
This extends to our two-dimensional signal values:
For $x, y \in \{(0,0), (1,0)\}$ it holds that $x \lor y = \max\{x,y\}$, where $\max$ denotes a component-wise maximum.
This motivates the MAX gadget, shown in Figure \ref{fig:knight:max}, which precisely implements a component-wise maximum.
\begin{figure}
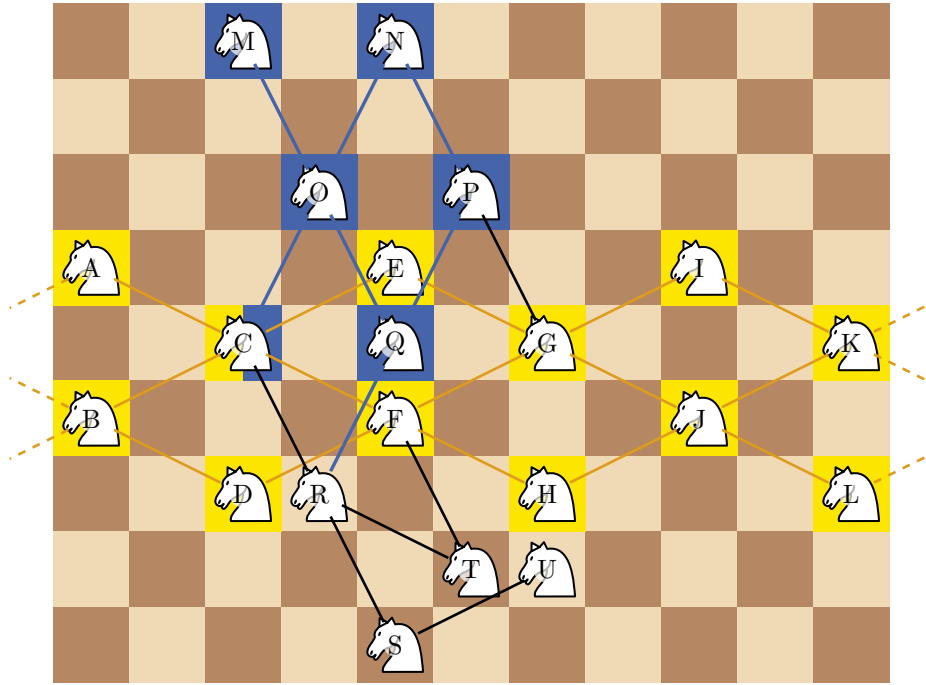

    \centering
    \knightAddition
    \caption{The MAX gadget computes the component-wise maximum of the blue and the yellow signal.}
    \label{fig:knight:max}
\end{figure}
The idea of the gadget is that the yellow signal is propagated as usual while the blue signal captures into the yellow wire.
When a component of the blue signal is 1, this allows a blue or white 2-knight to remain, acting as an Increment gadget for the corresponding yellow signal component.
This way, a component of the output is 1 if it was set to 1 for at least one of the inputs, thereby computing the component-wise maximum.

\begin{lemma}
\label{lem:knight:max:positive}
    There is a capturing sequence such that the MAX gadget with inputs $(u,v)$ and $(w,x)$ outputs their component-wise maximum $(\max\{u, w\}, \max\{v, x\})$.
\end{lemma}

\begin{proof}
    Start by reducing the 2-antenna $(U,S,R)$ to a $0$-knight $R$ via captures $\threeSquareMove{U}{S}{R}$.
    
    Consider now the case of a blue $(0,0)$-signal, i.e., knight $M$ is absent, while $N$ has a budget of $0$.
    Then the capturing sequence $\threeSquareMove{P}{N}{O}, \threeSquareMove{Q}{O}{C}, \threeSquareMove{T}{R}{C}$ leaves behind just the yellow wire (with virtual $0$-knight $C$ having an actual budget of 0).
    Thus, under this capturing sequence the gadget yields the identity function on the yellow signal, as expected.

    For a blue input of $(1,0)$, knight $M$ is present and the capturing sequence $\threeSquareMove{P}{N}{O}, \threeSquareMove{M}{O}{C}, \threeSquareMove{Q}{R}{C}$ leaves behind a first-component Increment for the yellow signal in the form of 2-knight $T$.
    Conversely, for a blue input of $(0,1)$, knight $N$ has a budget of 1 and the capturing sequence $\twoSquareMove{N}{O}, \threeSquareMove{Q}{O}{C}, \threeSquareMove{T}{R}{C}$ leaves behind a second-component Increment in the form of 2-knight $P$.
    Finally, for a blue input of $(1,1)$, combining the two approaches through captures $\twoSquareMove{N}{O}, \threeSquareMove{M}{O}{C}, \threeSquareMove{Q}{R}{C}$ leaves behind both a first- and a second-component Increment in 2-knights $T$ and $P$.
\end{proof}

It remains to show that no larger outputs are possible.
As the proof is slightly technical, we give it in Appendix \ref{appendix:knight}.
Here we state without proof:

\begin{restatable}{lemma}{knightMaxCorrect}
\label{lem:knight:max:correct}
    Given inputs $(v,w) \in \{0,1\}^2$ and $(x,y) \in \{0,1\}^2$ the MAX gadget computes their component-wise maximum \((\max\{v,x\}, \max\{w,y\}) \in \{0,1\}^2\).
\end{restatable}

\begin{corollary}
    Restricted to inputs $(x,0), (y,0) \in \{(0,0), (1,0)\}$ the MAX gadget computes the OR function.
\end{corollary}

\subsection{The AND Gadget}
The AND gadget has two inputs, $(x,0)$ and $(y,0)$ and one output, $(x \land y, 0)$.
Shown in Figure \ref{fig:knight:composite:and}, it is composed of multiple smaller gadgets connected by wire.
\begin{figure}
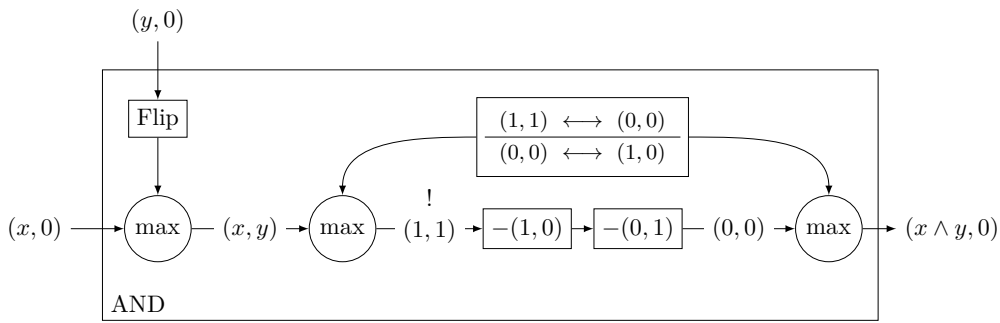

    \centering
    \scalebox{0.9}{\knightAndGadget}
    \caption{The AND gadget computes the logical conjunction of its two inputs.}
    \label{fig:knight:composite:and}
\end{figure}
Of these we have already seen the Decrement and MAX gadget.
We now introduce the remaining two building blocks.

\subparagraph*{Asymmetric Value Production Gadget (AVP)}
Placed at the top right of the AND gadget lies the AVP gadget, another value production gadget.
Recall that the VAR gadget combined a 3-VAL and a Decrement gadget to produce outputs $(1,0) \leftrightarrow (0,0)$ or $(0,0) \leftrightarrow (1,0)$.
Similarly, the AVP gadget produces outputs $(1,1) \leftrightarrow (0,0)$ or $(0,0) \leftrightarrow (1,0)$, by combining a 3-VAL and a Zeroing gadget (see Figure \ref{fig:knight:asym-var}).
\begin{figure}
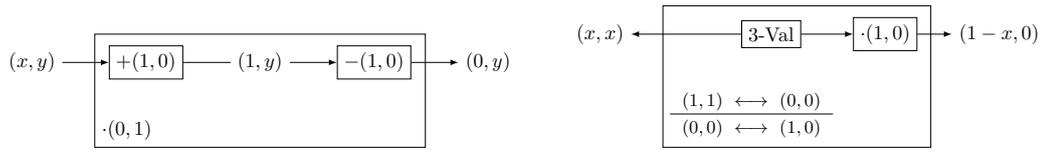

    \centering
    \scalebox{0.75}{
        \knightZeroComponent
    }
    \hfill
    \scalebox{0.75}{
        \knightAsymmetricVar
    }
    \caption{Left: The Zeroing gadget sets a component of a signal to 0. Right: Combining the Zeroing gadget with a 3-VAL gadget yields an AVP gadget.}
    \label{fig:knight:asym-var}
\end{figure}
This Zeroing gadget sets a component of a signal to 0 regardless of its previous value.
Consisting of an Increment followed by a Decrement gadget, its correctness follows directly from the correctness of its parts.

To see the correctness of the AVP gadget, we consider the different possible outputs of its 3-VAL gadget.
For 3-VAL outputs $(0,0) \leftrightarrow (1,1)$ the gadget produces an overall output of $(0,0) \leftrightarrow (1,0)$.
For 3-VAL outputs $(1,1) \leftrightarrow (0,0)$ or $(1,0) \leftrightarrow (0,1)$ the overall outputs are $(1,1) \leftrightarrow (0,0)$ and $(1,0) \leftrightarrow (0,0)$ respectively.
By monotony the second of these is not chosen.
Thus, this gadget correctly implements the asymmetric value production.

An analogous version of the gadget swaps the components of the outputs:
Placing a $\cdot (0,1)$ gadget to the left of the 3-VAL gadget yields outputs $(0,0) \leftrightarrow (1,1)$ or $(0,1) \leftrightarrow (0,0)$ instead.

\subparagraph*{Flip Gadget}
The top left of the AND gadget houses the Flip gadget, which is unpacked in Figure \ref{fig:knight:composite:flip}.
\begin{figure}
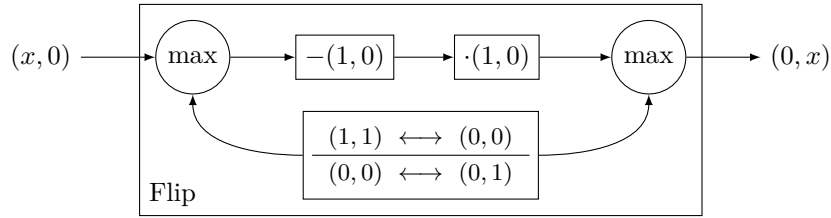

    \centering
    \knightFlipGadget
    \caption{The Flip gadget flips the two components of a signal.}
    \label{fig:knight:composite:flip}
\end{figure}
The Flip gadget computes the function that maps $(x,0) \in \{(0,0), (1,0)\}$ to $(0,x) \in \{(0,0), (0,1)\}$.
To show its correctness, we consider the different possible inputs.
Observe that the Flip gadget contains a Decrement gadget, enforcing that the output of the MAX gadget to its left has its first component set to 1.
Thus, for an input of $(0,0)$ to the Flip gadget, any clearing sequence evaluates the AVP gadget to $(1,1) \leftrightarrow (0,0)$.
In this case the signal undergoes the following transformations: $(0,0) \xrightarrow{\max} (1,1) \xrightarrow{-(1,0)} (0,1) \xrightarrow{\cdot (1,0)} (0,0) \xrightarrow{\max} (0,0)$.
If the input to the Flip gadget is $(1,0)$, the AVP gadget can evaluate to $(0,0) \leftrightarrow (0,1)$ leading to the following sequence: $(1,0) \xrightarrow{\max} (1,0) \xrightarrow{-(1,0)} (0,0) \xrightarrow{\cdot (1,0)} (0,0) \xrightarrow{\max} (0,1)$.
Overall this leads to the desired outputs.

\subparagraph*{Correctness of the AND Gadget}
Having understood its parts we now turn our attention to the AND gadget itself.
Its workings can directly be read off Figure \ref{fig:knight:composite:and}.
Its inputs $(x,0)$ and $(y,0)$ are combined to $(x,y)$ using a Flip and a MAX gadget.
The two Decrements lead to an error, \emph{unless} the signal going into them is a $(1,1)$-signal, in which case they produce a $(0,0)$-signal.
Thus, if $(x,y) \neq (1,1)$, the AVP gadget produces a $(1,1)$-signal to the left, yielding an overall AND gadget output of $(0,0)$.
If $(x,y) = (1,1)$, the AVP instead produces a $(0,0)$-signal to the left, yielding an overall output of $(1,0)$.
We have seen:
\begin{lemma}
    The AND gadget computes the function that maps inputs $(x,0), (y,0) \in \{(0,0), (1,0)\}$ to $(1,0)$ if both $x$ and $y$ are 1, and to $(0,0)$, otherwise.
\end{lemma}

\bibliography{lipics-v2021-sample-article}

\appendix
\section{King}
\label{appendix:king}

\subsection{NP-hardness of King \solo}
\label{appendix:correctness}
\canonicalsequence*
\begin{proof}
    Let $s$ be a clearing sequence.
    Since no VAR gadget contains the final square, the middle king of each VAR gadget captures two kings to one side in $s$ and one of the adjacent kings captures one king to the other side.
    It is easy to verify that a capturing sequence $s'$ obtained from $s$ by clearing the VAR gadgets first is still a clearing sequence.

    For the second property, we first show that after clearing all VAR gadgets, every output king is now a cut vertex in the resulting capture graph $C'$.
    We denote the original gadget graph by $G$ and the gadget graph after removing all VAR gadgets by $G'$.
    Note that in $G'$, each gadget except for the 1-TEST has exactly one outgoing edge.
    Consider the output king of some gadget $A$ that is not the 1-TEST.
    If we remove the output king, then gadget $A$ loses its outgoing edge in $G'$.
    We show that gadget $A$ and the 1-TEST gadget are not connected by an undirected path in $G'$.
    If $A$ is the only neighbor of the 1-TEST gadget in $G$, then we are done.
    Otherwise, assume that there is an undirected path $P$ from gadget $A$ to the 1-TEST gadget in $G'$.
    Following the edges in $P$, each edge is reversed compared to the corresponding edge in $G'$ or non-reversed.
    Since the 1-TEST gadget has no outgoing edges, the last edge in $P$ is non-reversed.
    Since $A$ has no outgoing edges anymore, the first edge in $P$ is reversed.
    Consider the first time in $P$ where a non-reversed edge follows directly after a reversed edge.
    Let $B$ be the gadget that is shared by both edges.
    But then, $B$ has two outgoing edges, which is a contradiction.
    Thus, there is no undirected path from $A$ to the 1-TEST gadget in $G'$, and the output king of $A$ is a cut vertex in $C'$.

    By Lemma~\ref{lem:0-piece-properties}, this means that each output king is a virtual $0$-king.
    Thus, a king on an output square can only move if the previous gadget is cleared.
    Note that each output king is also an input king for some gadget.
    Consider a gadget $A$ and all moves in $s'$ within $A$ that do not involve an input king of $A$.
    Then, these moves can be postponed until all previous gadgets are cleared.
    Moreover, if a king in $A$ captures an input square, then it is never worse to postpone this move until all previous gadgets are cleared.
    This means we can reorder the moves such that gadgets are cleared in topological order, and that moves within the same gadget form a consecutive capturing sequence.

    Since no output king moves before the corresponding gadget is cleared, the clearing sequence is a concatenation of normalized capturing sequences.
\end{proof}

\kingnphard*
\begin{proof}
    We show that there is a clearing sequence of $I_\Phi$ if and only if $\Phi$ is satisfiable.
    Let $\phi$ be a truth assignment that satisfies $\Phi$.
    We construct a clearing sequence for $I_\Phi$.
    Consider a variable $x$.
    If $x$ is set to true, then the middle king of each VAR gadget for $x$ captures twice towards the negative side.
    Otherwise, they capture once to the positive side.
    Thus, each variable gadget produces one $0$-signal and one $1$-signal.
    We clear the remaining gadgets in a topological order.
    Assuming the correctness of the wire, OR-, AND, and X gadget, it is possible to clear each gadget with the expected output.
    Since $\phi$ evaluates to true, the input square of the 1-TEST gadget contains a $1$-king after clearing all kings except for the 1-TEST gadget.
    The $1$-king then captures the next king, and the last king of the gadget captures twice, clearing the instance.

    Now, consider a clearing sequence of $I_\Phi$.
    By Lemma~\ref{lemma:canonical-sequence}, we may assume that the clearing sequence is canonical.
    We can now fix a truth assignment of $\Phi$.
    Let $x$ be a variable.
    There are at most two VAR gadgets corresponding to $x$.
    If in all such VAR gadgets, the middle king captures towards the negative side, then we set $x$ to true.
    Otherwise, we set it to false.

    It is clear that the final square of every clearing sequence is the second king in the 1-TEST gadget.
    Moreover, clearing the 1-TEST gadget is only possible if the input square contains a $1$-king after all other gadgets are cleared.
    As a canonical clearing sequence exactly mimics the evaluation of the given \sat-instance with its embedding, it follows directly from the correctness of the gadgets that the chosen truth assignment satisfies $\Phi$.

    As the reduction is polynomial, King \solo is \NP-hard.
\end{proof}

\subsection{Crossing Gadget}
\label{appendix:crossing}
\subparagraph*{IN-B Gadget} Figure~\ref{king:checkb} shows the IN-B gadget (see Table~\ref{table:crossing} for the IN-B-function).
The upper input corresponds to the output of the IN-A gadget.
The left input is signal $b$.
\begin{figure}
    \centering
    \resizebox{0.99\textwidth}{!}{
    \begin{newboard}{7}{6}
        \draw[inputsquare] (1,4) rectangle ++(1,1);
        \draw[inputsquare] (4,1) rectangle ++(1,1);
        \draw[outputsquare] (4,7) rectangle ++(2,1);
        \kingnew[id]{4}{1}{0}
        \kingnew[id]{1}{4}{0}
        \gfunc[id]{1}{1}
        \draw[move] (I) -- (B) -- (L);
        \draw[move] (J) -- (L) -- (N);
        \draw[move] (O) -- (N) -- (R);
        \draw[move] (D) -- (A) -- (C);
        \draw[move] (E) -- (C) -- (G);
        \draw[move] (F) -- (G) -- (K);
        \draw[move] (H) -- (K) -- (M);
        \draw[move] (Q) -- (M) -- (P);
        \draw[move] (S) -- (P) -- (R);
    \end{newboard}
    \begin{newboard}{7}{6}
        \draw[inputsquare] (1,4) rectangle ++(1,1);
        \draw[inputsquare] (4,1) rectangle ++(1,1);
        \draw[outputsquare] (4,7) rectangle ++(2,1);
        \kingnew[id]{4}{1}{0}
        \kingnew[id]{1}{4}{1}
        \gfunc[id]{1}{1}
        \draw[move] (B) -- (L);
        \draw[move] (C) -- (A) -- (D);
        \draw[move] (H) -- (D) -- (F);
        \draw[move] (G) -- (F) -- (J);
        \draw[move] (E) -- (J) -- (L);
        \draw[move] (G) -- (F) -- (J);
        \draw[move] (I) -- (L) -- (N);
        \draw[move] (O) -- (N) -- (R);
        \draw[move] (K) -- (M) -- (P);
        \draw[move] (Q) -- (P) -- (R);
        \draw[move] (S) -- (R);
    \end{newboard}
    \begin{newboard}{7}{6}
        \draw[inputsquare] (1,4) rectangle ++(1,1);
        \draw[inputsquare] (4,1) rectangle ++(1,1);
        \draw[outputsquare] (4,7) rectangle ++(2,1);
        \kingnew[id]{4}{1}{1}
        \kingnew[id]{1}{4}{1}
        \gfunc[id]{1}{1}
        \draw[move] (B) -- (I);
        \draw[move] (A) -- (C);
        \draw[move] (D) -- (G) -- (K);
        \draw[move] (L) -- (I) -- (E);
        \draw[move] (J) -- (E) -- (C);
        \draw[move] (F) -- (C) -- (G);
        \draw[move] (H) -- (K) -- (M);
        \draw[move] (Q) -- (M) -- (P);
        \draw[move] (O) -- (P) -- (R);
        \draw[move] (N) -- (R);
    \end{newboard}
    \begin{newboard}{7}{6}
        \draw[inputsquare] (1,4) rectangle ++(1,1);
        \draw[inputsquare] (4,1) rectangle ++(2,1);
        \draw[outputsquare] (4,7) rectangle ++(2,1);
        \kingnew[id]{4}{1}{2}
        \kingnew[id]{1}{4}{0}
        \gfunc[id]{1}{1}
        \kingnew[id]{5}{1}{1}
        \draw[move] (L) -- (B) -- (I);
        \draw[move] (J) -- (I) -- (E);
        \draw[move] (F) -- (E) -- (C);
        \draw[move] (T) -- (C);
        \draw[move] (A) -- (C) -- (G);
        \draw[move] (D) -- (G) -- (K);
        \draw[move] (H) -- (K) -- (M);
        \draw[move] (Q) -- (M) -- (P);
        \draw[move] (O) -- (P) -- (R);
        \draw[move] (N) -- (R);
    \end{newboard}
    }
    \caption{Capturing sequences for the IN-B gadget.
    If both inputs are $0$, then $\threeSquareMove{D}{A}{C}, \threeSquareMove{E}{C}{G}, \threeSquareMove{F}{G}{K}, \threeSquareMove{H}{K}{M}, \threeSquareMove{Q}{M}{P}, \threeSquareMove{S}{P}{R}, \threeSquareMove{I}{B}{L}, \threeSquareMove{J}{L}{N}, \threeSquareMove{O}{N}{R}$ outputs a $0$-signal.
    If the left input is $1$, then $\threeSquareMove{C}{A}{D}, \threeSquareMove{H}{D}{F}, \threeSquareMove{G}{F}{J}, \threeSquareMove{E}{J}{L}, \twoSquareMove{B}{L}, \threeSquareMove{I}{L}{N}, \threeSquareMove{O}{N}{R}, \threeSquareMove{K}{M}{P}, \threeSquareMove{Q}{P}{R}, \twoSquareMove{S}{R}$ outputs a $1$-signal.
    If both inputs are $1$, then $\twoSquareMove{B}{I}, \threeSquareMove{L}{I}{E}, \threeSquareMove{J}{E}{C}, \twoSquareMove{A}{C}, \threeSquareMove{F}{C}{G}, \threeSquareMove{D}{G}{K}, \threeSquareMove{H}{K}{M}, \threeSquareMove{Q}{M}{P}, \threeSquareMove{O}{P}{R}, \twoSquareMove{N}{R}$, outputs a $2$-signal.
    If the upper input is $2$ and the left input is $0$, then $\threeSquareMove{L}{B}{I}, \threeSquareMove{J}{I}{E}, \threeSquareMove{F}{E}{C}, \twoSquareMove{T}{C}, \threeSquareMove{A}{C}{G}, \threeSquareMove{D}{G}{K}, \threeSquareMove{H}{K}{M}, \threeSquareMove{Q}{M}{P}, \threeSquareMove{O}{P}{R}, \twoSquareMove{N}{R}$, outputs a $2$-signal.}
    \label{king:checkb}
\end{figure}

\kinginb*
\begin{proof}
    Figure~\ref{king:checkb} shows normalized capturing sequences for the relevant combinations with the desired output.
    All other cases directly follow from monotony (Observation~\ref{obs:monotony}).

    On the other hand, we show for the inputs $(0, 0)$, $(0, 1)$ and $(1, 0)$ that it is impossible to have a better output.
    We begin with showing that for every normalized capturing sequence, there are at least eight virtual $0$-kings (without inputs).
    It is clear that one of the two output squares is a virtual $0$-king.
    Moreover, there are four pairwise disjoint vertex cuts $\set{I, L}, \set{E, J, N}, \set{G, H, M}$ and $\set{K, P, Q}$.
    By Lemma~\ref{lem:0-piece-properties}, each vertex cut contains at least one virtual $0$-king and all virtual $0$-kings are connected.
    We note that at least eight virtual $0$-kings are needed to satisfy these properties.
    This leaves us with at most nine virtual $2$-kings (without inputs).

    If the left signal is $1$ and the upper signal is $0$, the total virtual budget is $19$, and either $18$ kings need to be cleared (for a $0$- or a $1$-signal) or $17$ kings with a remaining budget of $3$ (for a $2$-signal).
    Thus, the best possible output is a $1$-signal.

    If the left signal is $0$ and the upper signal is $1$, we have the same total virtual budget as in the previous case.
    We show that there is no normalized capturing sequence that does not lose budget, and thus, that the gadget outputs $0$ for these inputs.
    For this, we prove that in every normalized capturing sequence, there is a set of virtual $2$-pieces that does not have sufficiently many virtual $0$-pieces in its neighborhood.
    We distinguish three cases, depending on whether king $L$ or king $M$ are virtual $0$-kings.
    Since they form a vertex cut, we know by Lemma~\ref{lem:0-piece-properties} that at least one of them is a virtual $0$-king.

    If king $L$ is not a virtual $0$-king, then kings $I$ and $M$ are virtual $0$-kings.
    For connecting all virtual $0$-kings while using at most eight virtual $0$-kings, king $P$ has to be virtual a $0$-king as well, and kings $N$ $O$ and $Q$ cannot be virtual $0$-kings.
    But then, $\set{N, O, Q, R, S}$ contain four virtual $2$-kings that are only connected to three virtual $0$-kings.
    By Lemma~\ref{lemma:lose-budget}, the clearing sequence loses budget.

    If king $M$ is not a virtual $0$-king, then king $Q$ and king $S$ are not virtual $0$-kings.
    Moreover, not both kings $K$ and $P$ can be virtual $0$-kings since otherwise, it is impossible to connect all virtual $0$-kings with eight virtual $0$-kings.
    If king $K$ is a virtual $0$-king, then $Q$ is a $2$-king without a virtual $0$-kings in its neighborhood.
    If king $P$ is a virtual $0$-king, then $\set{M, Q}$ are two virtual $2$-kings that only have one virtual $0$-king in their neighborhood.
    In all cases, budget is lost.

    If both kings $L$ and $M$ are virtual $0$-kings, then kings $N$ and $P$ are also virtual $0$-kings.
    Moreover, $\set{I, J, K}$ contain at least one virtual $0$-king as a vertex cut.
    But then, $\set{C, D, E, F, G, H}$ contain at least four virtual $2$-kings that only have three virtual $0$-kings in their neighborhood.
    Thus, budget is lost.

    By Observation~\ref{obs:monotony}, the best possible output for $x = 0$ and $y = 0$ is also $0$.

    The highest possible output value is $2$, thus the lemma statement holds trivially for all other inputs.
\end{proof}

\subparagraph*{OUT Gadget} Figure~\ref{fig:king:out} shows the OUT gadget (see Table~\ref{table:crossing} for the OUT-function).
The left input corresponds to the output of the IN-B gadget.
The upper input is signal $\neg \tilde{b}$.
\begin{figure}
    \centering
    \resizebox{0.8\textwidth}{!}{
    \begin{newboard}{7}{5}
        \draw[inputsquare] (1,5) rectangle ++(1,1);
        \draw[inputsquare] (4,1) rectangle ++(1,1);
        \draw[outputsquare] (3,7) rectangle ++(1,1);
        \kingnew[id]{4}{1}{0}
        \kingnew[id]{1}{5}{1}
        \hfuncrot[id]{0}{0}
        \draw[move] (E) -- (A) -- (D);
        \draw[move] (C) -- (D) -- (G);
        \draw[move] (F) -- (G) -- (H);
        \draw[move] (O) -- (L) -- (H);
        \draw[move] (I) -- (H) -- (K);
        \draw[move] (B) -- (M);
        \draw[move] (N) -- (K) -- (M);
        \draw[move] (J) -- (M) -- (P);
    \end{newboard}
        \begin{newboard}{7}{5}
        \draw[inputsquare] (1,5) rectangle ++(1,2);
        \draw[inputsquare] (4,1) rectangle ++(1,1);
        \draw[outputsquare] (3,7) rectangle ++(1,1);
        \kingnew[id]{4}{1}{0}
        \kingnew[id]{1}{5}{1}
        \hfuncrot[id]{0}{0}
        \kingnew[id]{1}{6}{2}
        \draw[move] (E) -- (A) -- (D);
        \draw[move] (C) -- (D) -- (G);
        \draw[move] (F) -- (G) -- (H);
        \draw[move] (O) -- (L) -- (H);
        \draw[move] (I) -- (H) -- (K);
        \draw[move] (B) -- (M);
        \draw[move] (J) -- (K) -- (M);
        \draw[move] (Q) -- (M) -- (P);
        \draw[move] (N) -- (P);
    \end{newboard}
    \begin{newboard}{7}{5}
        \draw[inputsquare] (1,5) rectangle ++(1,1);
        \draw[inputsquare] (4,1) rectangle ++(1,1);
        \draw[outputsquare] (3,7) rectangle ++(1,1);
        \kingnew[id]{4}{1}{1}
        \kingnew[id]{1}{5}{0}
        \hfuncrot[id]{0}{0}
        \draw[move] (A) -- (D);
        \draw[move] (E) -- (D) -- (G);
        \draw[move] (C) -- (G) -- (H);
        \draw[move] (F) -- (H) -- (L);
        \draw[move] (I) -- (L) -- (N);
        \draw[move] (J) -- (B) -- (M);
        \draw[move] (K) -- (M) -- (P);
        \draw[move] (O) -- (N) -- (P);
    \end{newboard}
    }
    \caption{Capturing sequences for the OUT gadget.
    If the upper input is $0$, and the left input is $1$, then $\threeSquareMove{E}{A}{D}, \threeSquareMove{C}{D}{G}, \threeSquareMove{F}{G}{H}, \threeSquareMove{O}{L}{H}, \threeSquareMove{I}{H}{K}, \threeSquareMove{N}{K}{M}, \twoSquareMove{B}{M}, \threeSquareMove{J}{M}{P}$ outputs a $0$-king.
    If the left input is $2$, then $\threeSquareMove{E}{A}{D}, \threeSquareMove{C}{D}{G}, \threeSquareMove{F}{G}{H}, \threeSquareMove{O}{L}{H}, \threeSquareMove{I}{H}{K}, \twoSquareMove{B}{M}, \threeSquareMove{J}{K}{M}, \threeSquareMove{Q}{M}{P}, \twoSquareMove{N}{P}$ outputs a $1$-king.
    If the upper input is $1$, and the left input is $0$, then $\twoSquareMove{A}{D}, \threeSquareMove{E}{D}{G}, \threeSquareMove{C}{G}{H}, \threeSquareMove{F}{H}{L}, \threeSquareMove{I}{L}{N}, \threeSquareMove{O}{N}{P}, \threeSquareMove{J}{B}{M}, \threeSquareMove{K}{M}{P}$ outputs a $0$-king.}
    \label{fig:king:out}
\end{figure}
\kingout*
\begin{proof}
    Figure~\ref{fig:king:out} shows normalized capturing sequences for the relevant combinations with the desired output.
    All other cases directly follow from monotony (Observation~\ref{obs:monotony}).

    We show that for every normalized capturing sequence, there are at least seven virtual $0$-kings (without inputs).
    It is clear that the output square is a virtual $0$-king.
    Moreover, there are five pairwise disjoint vertex cuts $\set{C, D, E}, \set{F, G}, \set{H, I}, \set{J, M}$, and $\set{L, N}$.
    By Lemma~\ref{lem:0-piece-properties}, each vertex cut contains at least one virtual $0$-king and all virtual $0$-kings are connected.
    At least seven virtual $0$-kings are needed to satisfy these properties.
    This leaves us with at most seven virtual $2$-kings (without inputs).

    If it both inputs are $0$, the total virtual budget is $14$, but $15$ kings need to be cleared, which is impossible.

    If both inputs are $1$, the total virtual budget is $16$.
    We show that every normalized capturing sequence with seven virtual $0$-kings loses budget, and thus the best possible output is $0$.
    For this, we prove that in every normalized capturing sequence, there is a set of virtual $2$-pieces that does not have sufficiently many virtual $0$-pieces in its neighborhood.
    We distinguish three cases, depending on whether king $K$ or king $L$ are virtual $0$-kings.
    Since they form a vertex cut, we know by Lemma~\ref{lem:0-piece-properties} that at least one of them is a virtual $0$-king.

    If king $K$ is not a virtual $0$-king, then the kings $L$, $M$ and $N$ must be virtual $0$-kings to connect all regions with seven virtual $0$-kings.
    But then, $\set{C, D, E, F, G, H, I, J, K, O}$ contains the remaining three virtual $0$-kings and seven virtual $2$-kings.
    The virtual $2$-kings only have six virtual $0$-kings (all except the output king) in their neighborhood.

    If king $L$ is not a virtual $0$-king, then the kings $H$, $K$ and $N$ must be virtual $0$-kings to connect all regions with seven virtual $0$-kings.
    Moreover, $\set{J, M}$ contains a virtual $0$-king as a vertex cut.
    But then, $\set{C, D, E, F, G, I, L, O}$ contains the remaining two virtual $0$-kings and six virtual $2$-kings.
    The virtual $2$-kings only have at most four virtual $0$-kings in their neighborhood.

    If both kings $K$ and $L$ are virtual $0$-kings, then both kings $H$ and $M$ are also virtual $0$-kings to connect all regions with seven virtual $0$-kings.
    But then, $\set{C, D, E, F, G, I, O}$ contains the remaining two virtual $0$-kings and five virtual $2$-kings.
    Moreover, king $N$ is not a virtual $0$-kings.
    Thus, the virtual $2$-kings only have four virtual $0$-kings in their neighborhood.

    To conclude, in all cases, budget is lost.
    For the remaining cases where both inputs are less than $2$, the best possible output is also $0$ by Observation~\ref{obs:monotony}.

\end{proof}
\section{Knight 2-\solo}
\label{appendix:knight}

In this section we prove the correctness of the MAX gadget.
\knightMaxCorrect*
We already outlined in Lemma \ref{lem:knight:max:positive} that there are capturing sequences yielding the correct output.
It remains to show that no larger outputs are possible.
To aid the proof, we show an analogue of Corollary \ref{cor:knight:inner-inner}.

\begin{figure}
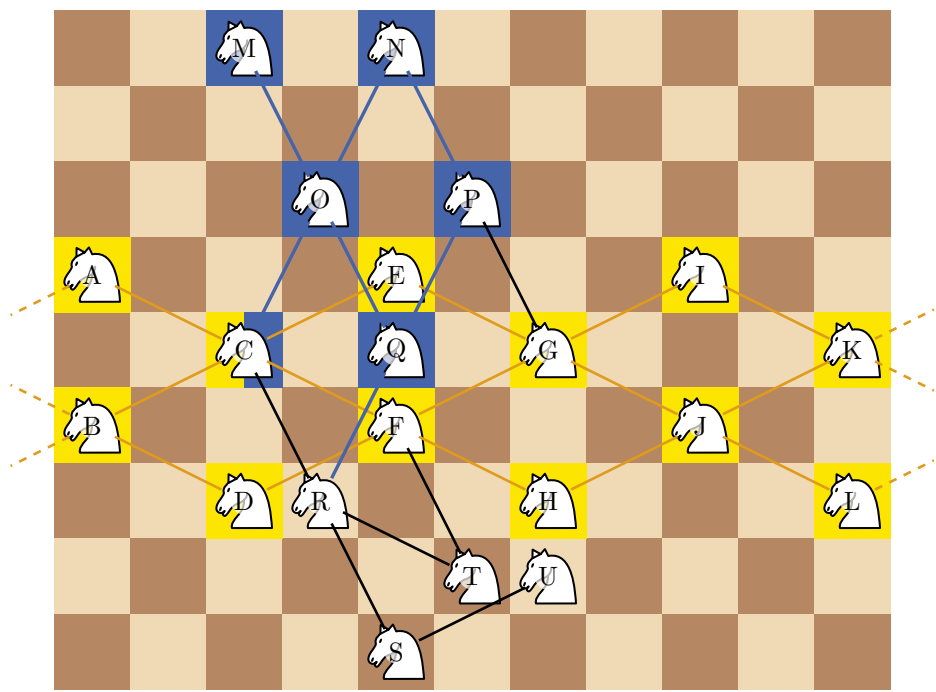

    \centering
    \knightAddition
    \caption{The MAX gadget computes the component-wise maximum of the blue and the orange signal.}
    \label{fig:appendix:knight:max}
\end{figure}

\begin{lemma}
\label{lem:knight:max:orange-wire}
    Let $s$ be a capturing sequence that clears the MAX gadget towards the right.
    Then each inner knight of the orange wire captures the inner knight to its right.
\end{lemma}

\begin{proof}
    To show the claim, we work our way from the left to the right, column by column.
    After cleaning up the antenna $(U,S,R)$ through captures $\threeSquareMove{U}{S}{R}$, we consider the left-most orange wire column.
    Being a standard wire column, it adheres to Lemma \ref{lem:knight:wire-virtual}, i.e., knight $B$ is a virtual $0$-knight, knight $A$ is a virtual $2$-knight and $B$ captures $C$, as it would become a stranded $0$-leaf after the alternative capture to $D$.
    
    Consider now the possible captures of knight $C$.
    It does not capture knight $E$, as the resulting $0$-leaf would be stranded.
    Assume for contradiction that it captures knight $R$.
    The resulting $0$-knight has two neighbors, knights $T$ and $Q$.
    After possible captures $\threeSquareMove{Q}{R}{T}$, the resulting $0$-leaf is stranded.
    This leaves $\threeSquareMove{T}{R}{Q}$, followed by $0$-knight $Q$ being propagated to $O$ or $P$.
    If $\threeSquareMove{P}{Q}{O}$, the resulting graph is disconnected as each knight of the cut $\{P, Q, C\}$ has already moved.
    After the alternative $\threeSquareMove{O}{Q}{P}$, both virtual $0$-knights $N$ and $P$ are cut vertices and the only way for the graph to stay connected is through captures $\threeSquareMove{N}{P}{G}$, which is impossible due to $N$ being a virtual $0$-knight, a contradiction.
    It follows that $C$ does not capture $R$.
    Next, assume for contradiction that knight $C$ captures knight $O$.
    Note that $0$-knight $R$ has only two neighbors remaining, knights $T$ and $Q$.
    Captures $\threeSquareMove{Q}{R}{T}$ would leave the resulting $0$-leaf stranded, thus, instead captures $\threeSquareMove{T}{R}{Q}$ are played.
    However, in the resulting configuration knight $P$ is a cut vertex and, thus, a virtual $0$-knight, alongside actual $0$-knights $O$ and $Q$.
    Then $0$-knight $Q$ does not have a virtual $2$-knight neighbor, a contradiction.
    It follows that knight $C$ does not capture knight $O$ either.
    The only remaining option is that knight $C$ captures knight $F$, as claimed.

    The remaining wire columns are straightforward.
    Knight $F$ does not capture knights $D$ or $H$, as the resulting $0$-leafs would be stranded, so it captures knight $G$ instead.

    Knight $G$ in turn does not capture knights $E$ or $I$.
    It does not capture knight $P$ either, as this knight $P$ would be disconnected from the right side of the gadget as each of $F$ and $G$ have already moved.
    Thus, knight $G$ captures knight $J$, as claimed, which finishes the proof.
\end{proof}

With these ingredients we turn to the main lemma.

\begin{proof}[Proof of Lemma \ref{lem:knight:max:correct}]
    Lemma \ref{lem:knight:max:positive} gives clearing sequences that yield the correct output.
    It remains to show that there is no clearing sequence yielding a larger output.
    Equivalently, we show that if a component of the orange signal is incremented, then the corresponding component in the blue signal was 1.
    
    Observe that each knight outside the orange wire captures towards an inner knight of the wire.
    This leaves two possibilities:
    A knight with budget less than 2 capturing into the orange wire does not modify the orange signal.
    A knight with budget 2 capturing into the orange wire acts as an Increment gadget of the respective component of the orange signal.
    
    Consider now the case that the first component of the orange signal is incremented.
    Then a dark-squared 2-knight captured an inner knight of the orange wire, which only leaves knight $T$.
    This leaves $0$-knight $R$ (after captures $\threeSquareMove{U}{S}{R}$) with two neighbors, knights $C$ and $Q$.
    By Lemma \ref{lem:knight:max:orange-wire}, knight $C$ captures knight $F$ so knight $R$ necessarily is propagated through moves $\threeSquareMove{Q}{R}{C}$ at some earlier point.
    Next, consider the virtual $0$-knight $N$ of the blue wire.
    It does not capture knight $P$ since without $Q$ being present, the resulting $0$-knight would be a leaf and, thus, stranded.
    Therefore, it captures knight $O$ to become a $0$-knight.
    Knight $O$ is not captured by $Q$ or $C$, thus, it is propagated through captures $\threeSquareMove{M}{O}{C}$.
    It follows that the first component of the blue signal was 1.

    The case for the second component of the orange signal being incremented is similar:
    The two candidates for a second-component Increment are knights $O$ and $P$.
    If either knight captures into the orange wire as a 2-knight, this leaves knight $N$ with just one neighbor below it.
    Thus, if knight $N$ was a $0$-knight, it would be stranded.
    It follows that $N$ is a $1$-knight, which implies that the second component of the blue signal was 1.
\end{proof}

\end{document}